\documentclass[journal]{IEEEtran}

\usepackage{amsmath,amssymb,amsthm}
\usepackage[T1]{fontenc}
\usepackage{graphicx}
\usepackage{booktabs}
\usepackage{pifont}
\usepackage{algorithm}
\usepackage{algorithmic}
\usepackage{cite}
\usepackage{url}
\usepackage{xcolor}

\newtheorem{assumption}{Assumption}
\newtheorem{definition}{Definition}
\newtheorem{proposition}{Proposition}
\newtheorem{theorem}{Theorem}
\newtheorem{lemma}{Lemma}
\newtheorem{corollary}{Corollary}
\newtheorem{remark}{Remark}

\newcommand{\fmax}{\ensuremath{f_{\max}}}
\newcommand{\FSS}{\ensuremath{\overline{\mathrm{FSS}}}}
\newcommand{\FSSp}{\ensuremath{\mathrm{FSS}}}
\newcommand{\DWPRm}{\ensuremath{\mathrm{DWPR}}}
\newcommand{\Drob}{\ensuremath{D^{\mathrm{rob}}}}
\newcommand{\TEI}{\ensuremath{\mathrm{TEI}}}
\newcommand{\Lf}{\ensuremath{\mathbf{L}_{f}}}
\newcommand{\lamtwo}{\ensuremath{\lambda_{2}(\Lf)}}
\newcommand{\cmark}{\ding{51}}
\newcommand{\xmark}{\ding{55}}

\begin{document}

\title{Adversary-as-Agents: A Co-Evolutionary Agent-Based Threat-Modelling
Framework for Wireless and Mobile Networks}

\author{Indrakshi~Dey,~\IEEEmembership{Senior~Member,~IEEE,}
        Mohamed~Khalfallah~Hassan,
        Sayanti~Ghosh,~\IEEEmembership{Member,~IEEE,}
        and~Nicola~Marchetti,~\IEEEmembership{Senior~Member,~IEEE}
\thanks{{Indrakshi Dey is with the Department of Computing and Mathematics, South East
Technological University, Waterford, X91 K0EK Ireland (e-mail: indrakshi.dey@setu.ie).
Mohamed Khalafalla Hassan is with the School of Computing,
National College of Ireland, Dublin, D01 K5Y6 Ireland (e-mail:
mohamed.KhalafallaHassan@ncirl.ie).
Sayanti Ghosh and Nicola Marchetti are with the Department of Electrical
and Electronic Engineering, Trinity College Dublin, Dublin, D02 PN40 Ireland
(e-mail: saghosh@tcd.ie; nicola.marchetti@tcd.ie). Supported in part by the US-Ireland Research and
Development Partnership RI-SFI-23/US/3924, and Research Ireland Grant
13/RC/2077\_P2.}}}

\maketitle

\begin{abstract}
Threat modelling for wireless and mobile networks is dominated by static,
catalogue-driven methods that fix the adversary in advance and never ask whether an
attack is feasible against the defence actually deployed; agent-based resilience
studies share this limitation, optimising a defender against an exogenous threat
profile. We instead endogenise the adversary. Network entities run a decentralised
consensus agent-based model (DC-ABM) defence, while an adaptive adversary population
co-evolves by allocating a bounded attack budget across vulnerability-mode partitions.
The equilibrium is a rankable, feasibility-grounded threat model expressed through
three structural metrics: degeneracy-weighted path robustness, trust-weighted
functional substitutability, and robust degeneracy. We prove that the DC-ABM defence
contracts below its Byzantine breakdown threshold and that the induced adversary
payoff is convex against that threshold, so the emergent attack drives a subset of
partitions to breakdown in priority order. We further prove that this attack
concentrates on the least substitutable partitions, that severity diverges as the
induced corrupted-mass fraction approaches breakdown, and that operators can shrink
the attack surface at deployment time along a closed-form
connectivity-substitutability exchange rate. A reference implementation covering
Byzantine faults, jamming, pilot contamination, eavesdropping, and correlated function
failure, benchmarked against state-of-the-art baselines, confirms each result. It
predicts consensus breakdown with a root-mean-square error of $0.014$, attains a mean
Kendall correlation of $0.34$ with independently measured consensus degradation where
catalogue-based baselines remain near zero, and cuts the average Threat Exposure Index
by $37\%$ under runtime adaptation. The equilibrium solver scales as $O(M\log M)$ in
the number of partitions.
\end{abstract}


\section{Introduction}
\IEEEPARstart{W}{ireless} and mobile networks are evolving into highly distributed,
autonomous systems that span virtualised radio access, edge intelligence, and
consensus-driven orchestration. Their control and data planes increasingly rely on
many interacting entities that reach agreement through local message exchange,
rather than on a single central controller. This shift enlarges the attack surface
in ways that conventional threat modelling struggles to capture
\cite{itu2160,shostack}. It also changes the fundamental question that a threat
model must answer. In a distributed network the relevant question is no longer only
which assets an attacker can reach, but whether the network can absorb an attack by
substituting a compromised function with an equivalent one and still reach a correct
collective decision. This is a networking-layer question about topology, redundancy,
and the resilience of the distributed protocol, and it is exactly the question that
catalogue-based methods leave unanswered.

Established methodologies enumerate threats from a static catalogue and score them
with expert-assigned weights. The Spoofing, Tampering, Repudiation, Information
disclosure, Denial of service, and Elevation of privilege taxonomy, known as STRIDE
\cite{shostack}, attack trees \cite{schneier}, and logic-based attack-graph
analysers such as the Multihost, multistage Vulnerability Analysis tool, known as
MulVAL \cite{mulval}, all take this form, and the resulting threats are commonly
scored through the Common Vulnerability Scoring System, known as CVSS \cite{cvss}.
Two structural weaknesses follow. First, the adversary is exogenous. The catalogue
is fixed in advance and does not react to the defence that the operator actually
deploys, so it cannot represent an attacker that reallocates effort once it observes
where the network is hardened. Second, the analysis is feasibility-blind. It does
not ask whether an attack can be absorbed by the redundant, distributed structure of
the network, so it systematically over-rates or under-rates threats in architectures
such as cell-free and virtualised deployments that carry substantial functional
redundancy \cite{cellfree}. A threat that a catalogue rates as critical may be
harmless because the network reroutes around it, and a threat that a catalogue
ignores may be the one that pushes a distributed consensus past its breakdown point.

Agent-based modelling represents network entities as autonomous, locally
interacting agents whose collective behaviour produces emergent system-level
outcomes. It is a natural tool for distributed wireless systems and has been applied
to distributed resource management \cite{ivoghlian} and to Byzantine-resilient
consensus \cite{suvaidya}. Yet existing agent-based resilience studies model only
the defensive half of the interaction. They assume a given threat, such as a fixed
fraction of Byzantine agents or a fixed salience profile over attack modes, and they
then optimise a defender against it. The threat model itself is never an output of
the analysis. Agent-based ideas have also been applied directly to security, for
example to simulate cyber-warfare between attacker and defender agents
\cite{kotenko}, {while catalogue-based threat modelling has matured through structured
frameworks such as the MITRE ATT\&CK (Adversarial Tactics, Techniques,
and Common Knowledge) knowledge base \cite{mitre}}, and through STRIDE-based analyses of cyber-physical systems
\cite{strideCPS}. These efforts either replay fixed adversary playbooks or enumerate
static catalogues. None couples an adaptive adversary to a feasibility-grounded
distributed defence so that the ranking of threats emerges from the interaction
between that adversary and the defence it actually faces.
{This paper closes that gap by endogenising the adversary. 

We cast threat modelling
as a co-evolutionary agent-based interaction. The network runs a decentralised
consensus agent defence, which we call the decentralised consensus agent-based model
(DC-ABM), and an adaptive adversary population allocates a bounded budget over
vulnerability-mode partitions. We adopt DC-ABM because it provides a fully
distributed defence mechanism with adaptive trust-weighted consensus, exposes an
explicit resilience breakdown threshold, and admits analytical characterisation of
the relationship between network topology, functional substitutability, and
Byzantine resilience. These properties make it a suitable foundation for
feasibility-grounded threat modelling while remaining compatible with large-scale
wireless and mobile networks. The equilibrium of this interaction is the threat
model. It is an emergent, feasibility-grounded, and rankable object, rather than a
hand-curated list. To keep the framework within the networking layer rather than
the physical layer, the only cross-layer quantity we import is a scalar feasibility
signal, namely the target-to-achieved service-level ratio that any management plane
can observe. Everything else, including the structural metrics, the breakdown
threshold, and the equilibrium ranking, is computed from the network topology and
the consensus dynamics. Fig.~\ref{fig:model} summarises the resulting
architecture and is referred to throughout Sec.~\ref{sec:model}.}

The framework is positioned squarely within the concerns of the networking
community. {This perspective enables threat modelling to evolve alongside network adaptation rather than remaining fixed throughout the system lifecycle.}  The defence is a distributed consensus protocol whose resilience is
governed by the algebraic connectivity of the coordination graph, a topological
design variable. The threat ranking is feasibility-grounded through functional
substitution, which is a statement about routing and redundancy in the network. The
central control result shows how the two deployment-time levers available to an
operator, namely raising the algebraic connectivity of the coordination graph by
adding coordination links, and raising functional substitutability by reassigning
substitutable functions across entities, reshape the attack surface, and it prices
one lever against the other. {The solver that produces the ranking runs in $O(M\log M)$ time
in the number of partitions, which keeps the analysis practical for large mobile
deployments with drifting conditions.} Our contributions are as follows.
\begin{itemize}
\item We give an adversary-as-agents co-evolutionary formulation over
vulnerability-mode partitions with an explicit attack budget, defended by a
self-contained DC-ABM defence whose consensus update we state in full
(Sec.~\ref{sec:model}).
\item We prove a contraction and convergence theorem for the DC-ABM defence below its breakdown threshold, we prove that the micro-tuning factor is unimodal, and we give the breakdown factorisation that separates topological from substitution effects (Sec.~\ref{sec:defence}).
\item We build a feasibility-grounded threat calculus on top of that defence,
consisting of a per-partition exploitability-severity score and an aggregate
Threat Exposure Index. Within this calculus we prove the following five results.
\begin{enumerate}
\item The adversary payoff is convex against the breakdown barrier, so the
optimal attack is a vertex allocation rather than a smooth spread.
\item The emergent attack concentrates on the least substitutable partitions.
\item Severity diverges as the induced corruption approaches breakdown.
\item A Stackelberg configuration of the operator problem exists.
\item The operator can shrink the emergent attack surface at design time along a
closed-form exchange rate.
\end{enumerate}
We close the section with a Lipschitz sensitivity bound and an $O(M\log M)$
complexity guarantee (Sec.~\ref{sec:coevo}).
\item We validate every analytical result with a reference implementation,
specified in Sec.~\ref{sec:impl}, that integrates the agent-level consensus and
trust dynamics of equations \eqref{eq:consensus} and \eqref{eq:trust} directly, extracts the
structural coordinates from the converged agent state, and closes the loop with
the priority-fill equilibrium solver of Proposition~\ref{prop:conc}. We compare
against several state-of-the-art threat-modelling and defence baselines across
five representative wireless and mobile threats (Sec.~\ref{sec:results}).
\end{itemize}

\section{Related Work}\label{sec:related}
\subsection{Catalogue-Based and Graph-Based Threat Modelling}
The dominant threat-modelling methods build a catalogue of threats and score it.
STRIDE organises threats into six classes and drives a checklist analysis
\cite{shostack}. Attack trees decompose a goal into conjunctive and disjunctive
subgoals and propagate leaf costs to the root \cite{schneier}. Attack-graph tools
such as MulVAL reason over reachability and configuration to enumerate multistage
paths \cite{mulval}, and Bayesian attack graphs add probabilistic dependence between
exploitation steps \cite{bayesag}. CVSS attaches severity scores drawn from expert
judgement \cite{cvss}. The MITRE matrix systematises adversary tactics and
techniques into a shared vocabulary \cite{mitre}, and topological vulnerability
analysis composes reachability and configuration data into network-wide attack
graphs \cite{jajodia}. All of these represent the
adversary as an exogenous input. The threat set and its scores are fixed before the
defence is considered, so the analysis cannot express an attacker that reallocates
effort in response to hardening, and it cannot express whether a distributed network
can absorb a given attack by substituting the functions it compromises with
structurally equivalent functions hosted elsewhere in the network. Our framework differs in both
respects. The adversary is a budget-constrained optimiser whose allocation emerges
from its interaction with the deployed consensus defence, and every partition
carries structural coordinates that make its criticality a computed quantity
rather than an assigned one.

\subsection{Game-Theoretic Security}
Game theory models security as an interaction between strategic attackers and
defenders \cite{gametheory}. Stackelberg security games, in which a defender commits
to a strategy that an attacker then best-responds to, have been applied widely to
resource allocation for protection \cite{tambe,alpcanbasar}. Our formulation shares
the leader-follower structure, since the operator commits to a topology and a
tuning, and the adversary best-responds with a budget allocation. The difference is
that our follower problem is grounded in the feasibility of a specific distributed
defence: the attacker payoff is not an abstract utility but the realised proximity
of each partition to the consensus breakdown threshold, which we derive from the
dynamics rather than assume. This gives the equilibrium a physical meaning that a
generic security game does not carry.

\subsection{Resilient Distributed Consensus}
The defence we adopt belongs to the literature on Byzantine-resilient consensus.
Classical results establish that agreement among distributed agents can tolerate a
bounded fraction of arbitrarily faulty participants \cite{lamport}. Mean-subsequence-
reduced and trimmed-mean filters discard extreme neighbour values and admit
resilience guarantees on graphs with sufficient connectivity
\cite{leblanc,sundaram}. Fault-tolerant multi-agent optimisation extends these ideas
to distributed objectives \cite{suvaidya}. Approximate agreement under faults \cite{dolev}, iterative Byzantine consensus on
directed graphs \cite{vaidyabyz}, resilience to locally bounded adversaries
\cite{zhang}, and scaled consensus \cite{roy} extend the same filtering principle
in different directions. The resilience of such protocols is controlled by
graph-theoretic quantities \cite{mesbahi}, in particular the algebraic
connectivity \cite{fiedler} that governs the rate of information diffusion
\cite{olfatisaber,xiaoboyd} and the robustness notions that govern tolerance to
adversarial nodes \cite{leblanc}.
We take a trust-weighted trimmed-mean consensus as the defence, and rather than
proving a new resilience guarantee for it, we extract from it the operating
characteristic that a threat model needs, namely the breakdown threshold as a
function of connectivity and substitutability, and the contraction of the honest
disagreement below that threshold.

\subsection{Agent-Based and Learning Approaches}
Agent-based modelling has been used to simulate cyber-conflict between attacker and
defender populations \cite{kotenko} and to manage distributed wireless resources
with multi-agent reinforcement learning \cite{ivoghlian}. Learning red teams can
discover strong attacks by interacting with a defended environment. These approaches
are complementary to ours. They produce attacks or policies but not a ranked,
feasibility-grounded, and interpretable threat model with closed-form design
levers, such as the connectivity-substitutability exchange rate of
\eqref{eq:exchange} and the interior trimming and trust optimum of
Lemma~\ref{lem:kappa}, that tell an operator what to change and by how much.
Recent advances have also explored large language models (LLMs) and autonomous AI
agents for cyber-defence planning, attack simulation, vulnerability analysis, and
security reasoning \cite{pentestgpt,fangllm,pentestai,llmsecsurvey}. These approaches provide semantic reasoning and automated decision support but generally operate over predefined attack knowledge, textual descriptions, or historical security artefacts rather than explicitly modelling the underlying network dynamics. Consequently, they do not capture the interaction between distributed resilience mechanisms and adaptive adversaries that determines whether attacks remain operationally feasible. The proposed framework is complementary, providing a mathematically grounded network-level model that can serve as a structural reasoning layer beneath future AI-driven security agents.
Our equilibrium is analytically characterised, which is what allows the controllability
and exchange-rate results, and the priority-fill solver that produces it is orders of
magnitude cheaper than training a learning red team. We include a multi-agent
reinforcement learning red team and a defence based on the alternating direction
method of multipliers \cite{reifert} among our baselines in Sec.~\ref{sec:results}.

\section{System, Adversary, and Threat Surface}\label{sec:model}
Fig.~\ref{fig:model} summarises the framework developed in this section and the
next two: a population of defender agents running the DC-ABM consensus, a
population of adversary agents spending a bounded budget, the vulnerability-mode
partitions over which the two populations meet, the structural metrics through
which their interaction is read, and the macro and micro levers with which the
operator reshapes it.

\subsection{Network as Defender Agents}
Let $\mathcal{K}=\{1,\dots,K\}$ index the $K$ network entities, which may be nodes,
virtualised functions, or slices. The entities are interconnected by a coordination
graph $\mathcal{G}_{f}$, which is the graph over which the management plane runs
its consensus-based control exchanges \cite{vestin} and which may be regular,
small-world \cite{wattsstrogatz}, or scale-free \cite{barabasi}. We write $\Lf$
for its graph Laplacian and $\lamtwo$ for its algebraic connectivity, the
second-smallest eigenvalue of $\Lf$, which governs the rate of information
diffusion across the network. Each entity $k$ holds a state
$x_{k}(t)\in\mathbb{R}$ that represents its local view of a quantity on which the
network must agree, for example a resource-allocation decision, a slice
configuration, or a trust estimate. Each entity runs a decentralised consensus update
that aggregates the reports of its neighbours with a trust-weighted trimmed mean. Let
$\mathcal{N}_{k}$ denote the neighbours of $k$, let $\tau_{kj}(t)\in[0,1]$ denote the
reputation that $k$ assigns to neighbour $j$, and let $\beta\in(0,1/2)$ denote the
trim fraction that discards the most extreme reports on each side. The update is
\begin{equation}
x_{k}(t+1)=\frac{\sum_{j\in\mathcal{R}_{k}(t)}\tau_{kj}(t)\,x_{j}(t)}
{\sum_{j\in\mathcal{R}_{k}(t)}\tau_{kj}(t)},
\label{eq:consensus}
\end{equation}
where $\mathcal{R}_{k}(t)\subseteq\mathcal{N}_{k}\cup\{k\}$ is the retained set that
remains after sorting the values $\{x_{j}(t):j\in\mathcal{N}_{k}\cup\{k\}\}$ and
removing the $\lfloor\beta\,|\mathcal{N}_{k}\cup\{k\}|\rfloor$ largest and smallest.
The reputations evolve so that persistently inconsistent neighbours lose influence,
\begin{equation}
\tau_{kj}(t+1)=\tau_{kj}(t)\,\big(1-\eta_{\tau}\,\phi\big(|x_{j}(t)-x_{k}(t+1)|\big)\big),
\label{eq:trust}
\end{equation}
where $\eta_{\tau}\in(0,1)$ is a trust learning rate and $\phi(\cdot)\in[0,1)$ is a
bounded, non-decreasing penalty, taken here as $\phi(z)=\tanh(z)$. Reputations are
clipped to a small positive floor so that no neighbour is permanently excluded on the
basis of transient disagreement.

As a management-plane observable we retain the scalar feasibility signal
\begin{equation}
D_{k}=\frac{s_{k}^{\mathrm{tar}}}{s_{k}^{\mathrm{ach}}},
\label{eq:feas}
\end{equation}
the ratio of the target service level $s_{k}^{\mathrm{tar}}$ to the achieved service
level $s_{k}^{\mathrm{ach}}$ of entity $k$. Entity $k$ is feasible if and only if
$D_{k}\le 1$. This degeneracy ratio is deliberately transmission-agnostic, so that
the framework remains a networking-layer construct and does not depend on the details
of the physical layer. The achieved service level degrades as the consensus is
disrupted, so $D_{k}$ is the bridge between the disagreement dynamics of
\eqref{eq:consensus} and the operational health of the network.

\subsection{Vulnerability-Mode Partitions}
Let $\Pi=\{\mathcal{P}_{1},\dots,\mathcal{P}_{M}\}$ partition the $M$ attackable modes
of the network. Representative wireless and mobile partitions are Byzantine or
compromised nodes, jamming and interference \cite{jamming}, pilot or identity
contamination, cache-aware eavesdropping, and correlated virtualised network
function (VNF) failure \cite{nfvsec}; together they cover the classes that
dominate published 5G threat analyses \cite{5gsec}.
Partitions are the atoms over which the adversary spends effort and over which the
operator assigns a salience profile $\{\zeta_{\mathcal{P}}\}$, with
$\zeta_{\mathcal{P}}\ge 0$ weighting the operational importance of partition
$\mathcal{P}$. These partitions subsume the categories of catalogue methods. Spoofing
and tampering map to pilot or identity contamination and to Byzantine nodes,
information disclosure maps to cache-aware eavesdropping, and denial of service maps
to jamming and to correlated function failure. MITRE tactics map onto the same five
modes: impact onto jamming and correlated VNF failure, since both degrade the service
that the consensus is meant to deliver, and defence evasion onto Byzantine or
compromised nodes and pilot or identity contamination, since both operate by passing
values that the trimming and trust filters of \eqref{eq:consensus}--\eqref{eq:trust}
do not reject \cite{mitre,strideCPS}. The difference is that
each partition here also carries the quantitative structural coordinates of
Definition~\ref{def:surface}, so that its criticality is computed rather than
assigned.

\subsection{Structural Metrics as Attack-Surface Coordinates}
For each partition $\mathcal{P}$ we use three structural metrics as the coordinates of
the attack surface. First, path robustness, the degeneracy-weighted path robustness
$\DWPRm(\mathcal{P})=1-\lVert\alpha_{\mathcal{P}}\rVert^{2}$, is the Gini-Simpson
spread of an entity's serviced resources across modes, where $\alpha_{\mathcal{P}}$ is
the rate-weighted mode-mass distribution over $\mathcal{P}$. A low value means that the
serviced resources are concentrated on few modes, so the partition is exposed, because
losing one mode cannot be compensated by the others. Second, trust-weighted functional
substitutability, the functional substitution score
$\FSSp^{\mathrm{trust}}(\mathcal{P})$, is the reputation-weighted count of structurally
distinct substitutes for the functions attacked in $\mathcal{P}$, and its network-wide
average is denoted $\FSS(\mathcal{P})$. A high value means that the network can reroute
around an attack on $\mathcal{P}$ by invoking equivalent functions elsewhere. Third,
robust degeneracy $\Drob(\mathcal{P})=\max_{a}\,s^{\mathrm{tar}}/s^{\mathrm{ach}}(a)$ is
the worst-case feasibility ratio over admissible adversary actions $a$, and it measures
how badly the achieved service can degrade under the strongest admissible attack on
$\mathcal{P}$.

\begin{definition}[Threat surface]\label{def:surface}
The threat surface is the map
$\Sigma:\ \mathcal{P}\mapsto\big(\DWPRm(\mathcal{P}),\,
\FSSp^{\mathrm{trust}}(\mathcal{P}),\,\Drob(\mathcal{P})\big)$ evaluated over all
$\mathcal{P}\in\Pi$. The pair $(\Sigma,b^{\star})$, where $b^{\star}$ is the adversary
equilibrium of Sec.~\ref{sec:coevo}, is the emergent threat model.
\end{definition}

\begin{figure}[!t]
\centering
\includegraphics[width=\columnwidth]{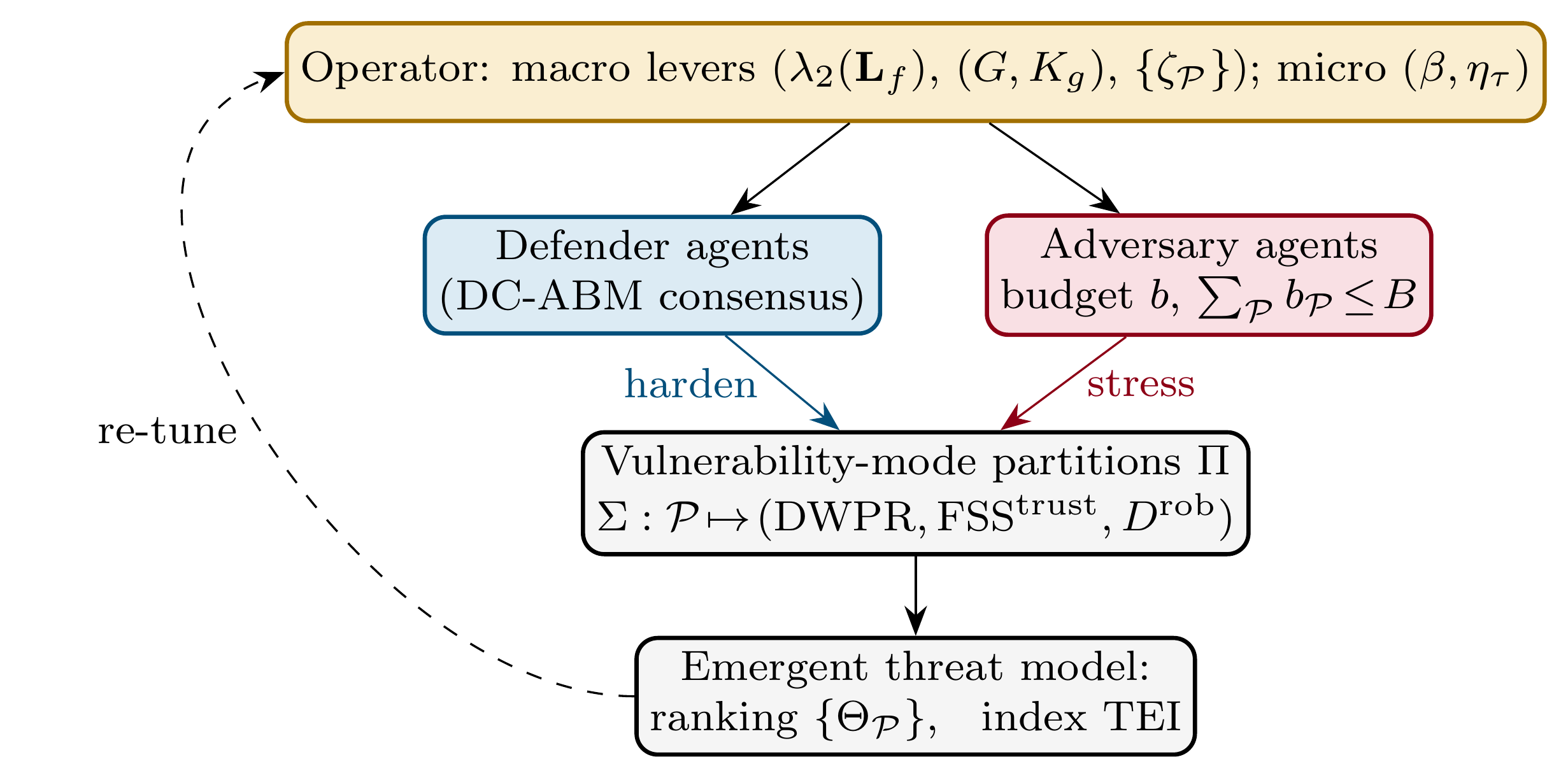}
\caption{Adversary-as-agents threat modelling. The defender and adversary agent
populations co-evolve over the vulnerability-mode partitions. Their interaction, read
through the structural metrics, emerges as a rankable threat model that the operator
can reshape through the macro levers, namely the coordination-graph topology and
the grouping of network functions into vulnerability-mode partitions, and through
the micro levers of trimming and trust.}
\label{fig:model}
\end{figure}

\section{Analysis of the DC-ABM Defence}\label{sec:defence}
This section establishes the operating characteristic of the DC-ABM defence that the
threat calculus of Sec.~\ref{sec:coevo} relies on. We show that below a breakdown
threshold the honest disagreement contracts to zero, we give the rate, we show that
the micro-tuning factor is unimodal in the trimming and trust parameters, and we
state the breakdown factorisation that separates the topological contribution from the
substitution contribution.

\subsection{Contraction and Convergence Below Breakdown}
Let $\mathcal{H}\subseteq\mathcal{K}$ denote the set of honest entities and let
$\mathcal{B}=\mathcal{K}\setminus\mathcal{H}$ denote the Byzantine entities, which may
send arbitrary and time-varying values. Let $f=|\mathcal{B}|/K$ denote the
corrupted-mass fraction. Write the honest disagreement as
$\Delta(t)=\max_{k\in\mathcal{H}}x_{k}(t)-\min_{k\in\mathcal{H}}x_{k}(t)$. We use the
standard notion of graph robustness for resilient consensus.

\begin{definition}[$r$-robustness \cite{leblanc}]\label{def:robust}
A graph is $r$-robust if for every pair of disjoint non-empty vertex subsets, at least
one of the two subsets contains a vertex with at least $r$ neighbours outside its own
subset. The graph is $(r,s)$-robust if in addition the required property holds for at
least $s$ vertices jointly across the two subsets.
\end{definition}

\begin{theorem}[Contraction and convergence]\label{thm:contraction}
Consider the trust-weighted trimmed-mean update \eqref{eq:consensus} with trim fraction
$\beta$ on a coordination graph that is $(2F+1)$-robust, where $F=\lceil\beta
K\rceil$ bounds the number of trimmed values per neighbourhood. If the corrupted-mass
fraction satisfies $f<\fmax$, where $\fmax$ is the largest fraction for which every
honest entity retains at least one honest neighbour after trimming, then there exists a
contraction factor $\rho(f)\in[0,1)$ such that
\begin{equation}
\Delta(t+1)\le\rho(f)\,\Delta(t),\qquad
\rho(f)=1-\gamma\,(1-f/\fmax),
\label{eq:rate}
\end{equation}
for a constant $\gamma\in(0,1]$ that depends on the connectivity and the trust floor.
Consequently the honest states converge geometrically to a common value, and
$\rho(f)\to 1^{-}$ as $f\to\fmax^{-}$.
\end{theorem}

\begin{proof}
The argument combines the safety of the trimmed mean with a diffusion bound. First,
consider safety. After sorting the values in the neighbourhood of an honest entity $k$
and removing the $F=\lceil\beta K\rceil$ largest of them and the $F$ smallest of
them, every retained value lies within the interval spanned by the honest values, because at most $F$ Byzantine values can lie
above the honest range and at most $F$ can lie below it, and the trim of size $F$ on
each side removes them before they can enter the average. This is the standard
safety property of the mean-subsequence-reduced filter on a $(2F+1)$-robust graph
\cite{leblanc,sundaram}. Hence the honest update in \eqref{eq:consensus} is a convex
combination of honest values, with weights given by the retained reputations, and no
honest state can leave the current honest interval. This gives
$\Delta(t+1)\le\Delta(t)$.

Next, consider strict contraction, which we quantify through the ergodicity
coefficient of the honest averaging operator rather than an aggregate diffusion
argument. By the safety property just established, the honest update can be written as
$x_{H}(t+1)=P(t)\,x_{H}(t)$, where $x_{H}$ collects the honest states and
$P(t)\in\mathbb{R}^{H\times H}$ is row-stochastic: its entry $P_{kj}(t)$ is the
normalised retained reputation that honest entity $k$ places on honest entity $j$ after
trimming, and the rows sum to one because every retained value lies in the honest
interval and the Byzantine contributions have been trimmed out. For any row-stochastic
$P$ the span seminorm $\Delta(x)=\max_i x_i-\min_i x_i$ obeys the classical Dobrushin
bound $\Delta(Px)\le\tau_{1}(P)\,\Delta(x)$, where
\begin{equation}
\tau_{1}(P)=1-\min_{k,\ell}\sum_{j}\min\!\big(P_{kj},P_{\ell j}\big)\in[0,1]
\label{eq:dobrushin}
\end{equation}
is the ergodicity (scrambling) coefficient \cite{seneta}. It remains to lower bound the
scrambling mass $\min_{k,\ell}\sum_{j}\min(P_{kj},P_{\ell j})$. Let $w_{\min}>0$ be the
smallest normalised retained reputation, which is bounded below because the reputations
are clipped to a positive floor $\underline{\tau}>0$. The $(2F+1)$-robustness of the
coordination graph guarantees that any two honest entities $k$ and $\ell$ retain at
least $c(f)\ge 1$ common honest neighbours after trimming whenever $f<\fmax$, since
fewer than the robustness margin of neighbours can be removed by the combined trim
and corruption. Here $c(f):=\min_{k,\ell\in\mathcal{H}}
|\mathcal{R}_{k}\cap\mathcal{R}_{\ell}\cap\mathcal{H}|$ denotes the smallest number
of honest entities that survive the trim simultaneously in the retained sets of any
two honest entities at corrupted-mass fraction $f$; it is the common support through
which information is guaranteed to flow between any two honest views.
Each such common neighbour contributes at least $w_{\min}$ to the inner minimum, so
$\min_{k,\ell}\sum_{j}\min(P_{kj},P_{\ell j})\ge w_{\min}\,c(f)$ and hence
$\tau_{1}(P(t))\le 1-w_{\min}\,c(f)$. The number of retained common honest neighbours
degrades linearly with the corrupted mass and vanishes at breakdown, so we write
$c(f)=c_{0}\,(1-f/\fmax)$ with $c_{0}>0$ fixed by the connectivity margin. Setting
$\gamma=w_{\min}\,c_{0}\in(0,1]$ gives the rate
$\rho(f)=1-\gamma\,(1-f/\fmax)$ in \eqref{eq:rate} with all constants identified
explicitly. As $f\to\fmax^{-}$ the common honest support shrinks to nothing, so
$c(f)\to 0$ and $\tau_{1}(P(t))\to 1^{-}$. In this limit the averaging operator
$P(t)$ ceases to be \emph{scrambling}, namely there is a pair of honest entities
$k,\ell$ whose retained sets share no honest neighbour, so the rows $P_{k\cdot}(t)$
and $P_{\ell\cdot}(t)$ have disjoint supports and
$\sum_{j}\min(P_{kj},P_{\ell j})=0$. No honest value is then common to the two
updates, the views of $k$ and $\ell$ are no longer mixed, and the Dobrushin bound
\eqref{eq:dobrushin} becomes vacuous because it no longer contracts the span
seminorm. Hence $\rho(f)\to 1^{-}$.
Geometric convergence of $\Delta(t)$ to zero for $f<\fmax$ follows by iterating
\eqref{eq:rate}.
\end{proof}

Theorem~\ref{thm:contraction} is the analytical backbone of the severity result in
Sec.~\ref{sec:coevo}, and its operational reading is the following. The defence
remains effective for every $f<\fmax$, but the per-step contraction factor $\rho(f)$
degrades smoothly and predictably towards one as the attack approaches breakdown.
Since the honest states are the values on which the entities act, the disagreement
$\Delta(T)$ that survives after a fixed operating horizon $T$ is what the network
carries into its resource-allocation decisions, and by \eqref{eq:rate} it grows like
$\rho(f)^{T}\Delta(0)$. A larger surviving disagreement means the entities serve
their traffic from inconsistent decisions, so the achieved service level
$s^{\mathrm{ach}}_{k}$ falls and the feasibility ratio $D_{k}$ of \eqref{eq:feas}
rises. In the limit $f\to\fmax^{-}$ we have $\rho(f)\to1^{-}$, the disagreement no
longer contracts at all, and $D_{k}$ diverges. The empirical contraction curves in
Sec.~\ref{sec:results} confirm both effects: a geometric decay of $\Delta(t)$ whose
measured per-step factor matches the $\rho(f)$ of \eqref{eq:rate}, and a stall of
that decay as $f$ approaches the threshold.

\subsection{Breakdown Threshold and Its Factorisation}

The contraction result of Theorem~\ref{thm:contraction} shows that the
breakdown threshold $\fmax$ depends on two qualitatively different mechanisms.
The first is the \emph{local consensus dynamics}, governed by the trimming
fraction $\beta$ and the trust adaptation rate $\eta_{\tau}$, which determine
how effectively the defender suppresses corrupted information. The second is
the \emph{network structure}, determined by the coordination topology and the
availability of functionally substitutable resources, which governs the amount
of redundancy available to tolerate attacks. Since these two mechanisms operate
at different decision scales and influence the consensus dynamics through
independent pathways, we model the breakdown threshold as a separable
macro--micro factorisation. This factorisation is motivated by the structure of
the consensus process and is subsequently validated through numerical
experiments rather than derived from first principles.
{\begin{assumption}[Breakdown factorisation]
\label{ass:break}
The DC-ABM breakdown threshold for partition $\mathcal{P}$ is modelled as
\begin{equation}
\fmax(\mathcal{P})
=
\underbrace{\kappa(\beta,\eta_{\tau})}_{\text{micro}}
\cdot
\underbrace{
\frac{g(\lamtwo)}
{1+c_{2}\left(1-\FSS(\mathcal{P})\right)}
}_{\text{macro}},
\label{eq:fmax}
\end{equation}
where $\kappa(\beta,\eta_{\tau})\in[0,\kappa^{\star}]$ captures the influence
of the runtime consensus parameters, $g(\lamtwo)$ is a monotone
non-decreasing function of the algebraic connectivity, and
$c_{2}>0$ quantifies the penalty associated with limited functional
substitutability.
The separable form reflects the modelling assumption that runtime adaptation
(micro scale) and structural redundancy (macro scale) contribute
approximately independently to the resilience margin. For $g$ we adopt the
saturating model
\[
g(\lamtwo)
=
\frac{\lamtwo}{1+c_{3}\lamtwo},
\]
which captures the diminishing marginal resilience obtained from increasing
network connectivity. The constant $c_{3}\ge 0$ is the saturation coefficient: it
fixes the connectivity scale $1/c_{3}$ beyond which additional links buy little
further resilience, and it is estimated once per deployment together with $c_{2}$
by fitting \eqref{eq:fmax} to measured thresholds. In the low-connectivity regime
$\lamtwo\ll1/c_{3}$ the model reduces to the linear form
$g(\lamtwo)\approx\lamtwo$, and $c_{3}=0$ recovers it exactly. Section~VI demonstrates that this factorised model
accurately predicts the measured breakdown threshold across diverse network
topologies and substitutability levels.
\end{assumption}}

{The factorised form separates deployment-level design decisions from runtime parameter adaptation, simplifying subsequent optimisation.} The micro factor $\kappa$ captures the effect of the trimming and trust parameters.
Trimming too little lets Byzantine values through, and trimming too much discards
honest neighbours and weakens the retained honest support, so there is an interior
optimum. The trust rate behaves similarly, since too small a rate fails to suppress
persistent attackers and too large a rate destabilises the reputations on transient
disagreement. The macro term captures the effect of topology and substitutability.
Higher algebraic connectivity raises the retained honest support and therefore the
tolerable corrupted mass, while higher substitutability lowers the substitution
penalty and again raises the threshold. The following lemma formalises the interior
optimum of the micro factor.

\begin{lemma}[Unimodality of the micro factor]\label{lem:kappa}
For fixed topology and partition, the breakdown threshold as a function of the trim
fraction $\beta$ and the trust rate $\eta_{\tau}$ is quasi-concave and attains an
interior maximum $\kappa^{\star}$ at a point $(\beta^{\star},\eta_{\tau}^{\star})$ with
$0<\beta^{\star}<1/2$ and $0<\eta_{\tau}^{\star}<1$.
\end{lemma}

\begin{proof}
Fix $\eta_{\tau}$ and vary $\beta$. At $\beta=0$ no values are trimmed, so a single
extreme Byzantine value enters every honest average and the tolerable corrupted mass is
zero. As $\beta$ increases from zero, each additional unit of trim removes one
potential Byzantine value from each side, which raises the tolerable corrupted mass, so
the threshold is increasing for small $\beta$. As $\beta$ approaches $1/2$, the trim
removes almost all neighbours, including honest ones, so the retained honest support
collapses and the threshold falls to zero. Hence the threshold is zero at both ends of
the interval and positive in the interior, and by the single-crossing property of the
marginal gain from trimming, which switches once from positive to negative as trimming
begins to remove honest rather than Byzantine values, the threshold is quasi-concave in
$\beta$ with a unique interior maximiser. The same argument applies to $\eta_{\tau}$,
since a zero rate fails to demote persistent attackers and a unit rate zeroes
reputations on transient disagreement, again giving zero threshold at both ends and a
positive interior. Quasi-concavity of the two-parameter threshold follows because the
two effects act on disjoint mechanisms, the retained set and the retained weights, and
their composition preserves the single interior maximiser. The maximiser lies strictly
inside the domain, which gives $\kappa^{\star}$ and $(\beta^{\star},\eta_{\tau}^{\star})$.
\end{proof}

Lemma~\ref{lem:kappa} guarantees that the micro levers have a well-defined optimum that
the operator can seek, and Sec.~\ref{sec:results} recovers this optimum empirically. We
now use the operating characteristic established here to build and analyse the
co-evolutionary threat model.

\section{Co-Evolutionary Threat Model}\label{sec:coevo}
\subsection{Adversary Agents and Induced Stress}
The adversary allocates a budget vector
$\mathbf{b}=(b_{1},\ldots,b_{M})$, where
$b_{\mathcal{P}}\ge0$ denotes the effort assigned to vulnerability
partition $\mathcal{P}$ and
$\sum_{\mathcal{P}} b_{\mathcal{P}}\le B$
is the total attack budget. To compare heterogeneous attack mechanisms on
a common resilience scale, we map the attack expenditure to an induced
Byzantine-equivalent corruption fraction,
$f_{\mathcal{P}}(b_{\mathcal{P}})$, which represents the effective
fraction of compromised consensus influence generated by investing budget
$b_{\mathcal{P}}$ against partition $\mathcal{P}$. Rather than modelling
the underlying physical attack process, this mapping serves as an abstract
resource-to-impact transformation that enables diverse attack types to be
analysed within a unified consensus-resilience framework.
We adopt the saturating exponential model
\begin{equation}
f_{\mathcal{P}}(b_{\mathcal{P}})
=
f_{\infty}
\left(1-e^{-a_{\mathcal{P}}b_{\mathcal{P}}}\right),
\label{eq:finduced}
\end{equation}
where $f_{\mathcal{P}}(0)=0$, $f_{\infty}$ denotes the maximum attainable
effective corruption level, and
$a_{\mathcal{P}}>0$ is a partition-specific attack-efficiency parameter.
This model captures the practical phenomenon of
\emph{diminishing marginal returns}: the initial attack budget
preferentially compromises the most vulnerable entities, whereas
additional expenditure must overcome progressively stronger protection,
redundancy, or functional diversity, resulting in a gradual saturation of
the achievable corruption level. The exponential mapping is
analytically attractive because it is smooth, strictly increasing,
naturally bounded by $f_{\infty}$, and avoids unrealistically unbounded
growth in the induced corruption level.
The subsequent analysis is not tied to the specific exponential form.
Rather, it requires only that
$f_{\mathcal{P}}(\cdot)$ be continuously differentiable,
monotonically increasing, and saturating. Hence, alternative mappings
that satisfy these properties, including logistic and
Michaelis--Menten functions, may be employed without affecting the
structural validity of the analytical results developed in the remainder
of the paper.

\subsection{Feasibility-Grounded Threat Score}
Using the breakdown threshold \eqref{eq:fmax}, we define for each partition an
exploitability term $E_{\mathcal{P}}$ and a severity term $S_{\mathcal{P}}$,
\begin{align}
E_{\mathcal{P}} &= \frac{1}{\FSSp^{\mathrm{trust}}(\mathcal{P})\,
\DWPRm(\mathcal{P})+\varepsilon},
\label{eq:expl}\\
S_{\mathcal{P}}(b_{\mathcal{P}}) &= \Drob(\mathcal{P})\,
\frac{f_{\mathcal{P}}(b_{\mathcal{P}})}{\fmax(\mathcal{P})},
\label{eq:sev}
\end{align}
where $\varepsilon>0$ is a small constant that prevents division by zero. Low
substitutability and low path robustness raise exploitability, since the network cannot
reroute around the attack, and worst-case degeneracy together with proximity to
breakdown raise severity. Consistent with Theorem~\ref{thm:contraction}, as
$f_{\mathcal{P}}\to\fmax(\mathcal{P})$ the achieved service collapses and the robust
degeneracy $\Drob$ diverges. We model this divergence explicitly as
\begin{equation}
\Drob(\mathcal{P})=\frac{D_{0}}{1-f_{\mathcal{P}}/\fmax(\mathcal{P})},
\label{eq:drob}
\end{equation}
where $D_{0}>0$ is the attack-free feasibility ratio of the partition, that is, the
value of $\Drob(\mathcal{P})$ at $b_{\mathcal{P}}=0$, so that $\Drob=D_{0}$ when the
partition is unattacked. The denominator of \eqref{eq:drob} is the residual honest
agreement of the partition, namely the normalised margin
$1-f_{\mathcal{P}}/\fmax(\mathcal{P})\in(0,1]$ that separates the induced corruption
from the breakdown threshold. It is exactly the factor that multiplies $\gamma$ in the
contraction rate \eqref{eq:rate}, and Theorem~\ref{thm:contraction} shows that it
vanishes, together with the contraction of $\Delta(t)$, as
$f_{\mathcal{P}}\to\fmax(\mathcal{P})^{-}$. The per-partition exploitability-severity score
$\Theta_{\mathcal{P}}$ and the aggregate Threat Exposure Index are
\begin{equation}
\Theta_{\mathcal{P}}(b_{\mathcal{P}})=E_{\mathcal{P}}\,S_{\mathcal{P}}(b_{\mathcal{P}}),
\qquad
\TEI(b)=\sum_{\mathcal{P}\in\Pi}\zeta_{\mathcal{P}}\,
\Theta_{\mathcal{P}}(b_{\mathcal{P}}).
\label{eq:tei}
\end{equation}
\vspace{-8mm}

\subsection{Bilevel Interaction and Its Equilibrium}
The operator sets the macro and micro configuration, and the adversary best-responds by
maximising the damage it can inflict for its budget,
\begin{equation}
b^{\star}=\arg\max_{b\ge 0,\ \mathbf{1}^{\top}b\le B}
\ \sum_{\mathcal{P}\in\Pi}U_{\mathcal{P}}(b_{\mathcal{P}}),
\qquad
U_{\mathcal{P}}:=\zeta_{\mathcal{P}}\,\Theta_{\mathcal{P}},
\label{eq:adv}
\end{equation}
where $\mathbf{1}$ is the all-ones vector. The emergent threat model is the pair
$(\Sigma,b^{\star})$, the surface coordinates ranked by the realised scores
$\{\Theta_{\mathcal{P}}(b^{\star}_{\mathcal{P}})\}$. We first establish that the
adversary problem is well posed.

\begin{proposition}[Convex payoff, breakdown barrier, and existence]\label{prop:exist}
Each payoff $U_{\mathcal{P}}$ is strictly increasing and strictly convex in
$b_{\mathcal{P}}$ on the feasible region where $f_{\mathcal{P}}<\fmax(\mathcal{P})$, and
$U_{\mathcal{P}}\to\infty$ as $f_{\mathcal{P}}\to\fmax(\mathcal{P})^{-}$. The adversary
problem is therefore well posed only against the breakdown barrier
$f_{\mathcal{P}}\le(1-\epsilon)\fmax(\mathcal{P})$, equivalently a per-partition budget cap
$b_{\mathcal{P}}\le\bar{b}_{\mathcal{P}}=f_{\mathcal{P}}^{-1}\big((1-\epsilon)\fmax\big)$.
Over the resulting compact feasible set a maximiser $b^{\star}$ exists, and because the
objective is convex it is attained at an extreme point of that set, that is, at an
allocation that drives a subset of partitions to their caps and leaves the rest at zero.
\end{proposition}

\begin{proof}
Fix a partition and write $u=f_{\mathcal{P}}/\fmax\in[0,1)$. Substituting
\eqref{eq:drob} into \eqref{eq:sev} gives severity $S_{\mathcal{P}}=D_{0}\,u/(1-u)$, so
$U_{\mathcal{P}}=\zeta_{\mathcal{P}}E_{\mathcal{P}}D_{0}\,h(u)$ with $h(u)=u/(1-u)$,
which is strictly increasing and diverges as $u\to 1^{-}$, giving the barrier at
breakdown. Differentiating through $u=f_{\mathcal{P}}(b)/\fmax$ with
$f_{\mathcal{P}}'(b)=a_{\mathcal{P}}(f_{\infty}-f_{\mathcal{P}})>0$ and
$f_{\mathcal{P}}''(b)=-a_{\mathcal{P}}f_{\mathcal{P}}'(b)$,
\[
U_{\mathcal{P}}''(b)=\frac{\zeta_{\mathcal{P}}E_{\mathcal{P}}D_{0}}{\fmax}
\,\frac{f_{\mathcal{P}}'(b)}{(1-u)^{2}}
\left[\frac{2\,f_{\mathcal{P}}'(b)}{\fmax(1-u)}-a_{\mathcal{P}}\right].
\]
The prefactor is positive, and the bracket equals
$a_{\mathcal{P}}\big[2(f_{\infty}-f_{\mathcal{P}})/(\fmax-f_{\mathcal{P}})-1\big]$. Since
the corruption efficiency saturates at a level above breakdown,
$f_{\infty}>\fmax$, we have $f_{\infty}-f_{\mathcal{P}}>\fmax-f_{\mathcal{P}}>0$ for all
feasible $f_{\mathcal{P}}<\fmax$, so the bracket exceeds $a_{\mathcal{P}}>0$ and
$U_{\mathcal{P}}''(b)>0$. Each payoff is thus strictly convex and increasing on its
feasible interval. This is the mathematically important correction to a naive
water-filling reading: the divergence of severity at breakdown makes the payoff convex,
not concave, so the problem is not a smooth interior optimisation but a convex
maximisation against the breakdown barrier. Imposing
$f_{\mathcal{P}}\le(1-\epsilon)\fmax$ yields the caps $\bar{b}_{\mathcal{P}}$ and the
feasible set $K=\{b:\,b\ge 0,\ \mathbf{1}^{\top}b\le B,\ b_{\mathcal{P}}\le
\bar{b}_{\mathcal{P}}\}$, which is compact and convex. The objective is continuous, so a
maximiser exists by the Weierstrass theorem. A convex function on a compact convex
polytope attains its maximum at an extreme point of the polytope \cite{boyd}, and the extreme points
of $K$ set each coordinate to $0$, to its cap $\bar{b}_{\mathcal{P}}$, or, for at most
one coordinate, to the residual budget. Hence the optimal attack fills a subset of
partitions to breakdown and leaves the rest untouched.
\end{proof}

\begin{remark}
The convex-barrier structure sharpens rather than weakens the concentration thesis. A
concave payoff would spread the budget smoothly across all partitions, whereas the
convex payoff drives the adversary to commit fully to the partitions it selects and to
select them in a strict priority order, which is exactly the vertex behaviour above.
When the budget exceeds the sum of caps the adversary can push every partition to its
barrier; the interesting regime is $B$ below that sum, where the choice of which
partitions to bring to breakdown is forced, and Proposition~\ref{prop:conc} identifies
that choice.
\end{remark}

\subsection{Emergent Threat Concentration}
\begin{proposition}[Emergent threat concentration]\label{prop:conc}
Let each partition have feasibility cap $\bar{b}_{\mathcal{P}}$ and cap efficiency
$\rho_{\mathcal{P}}=U_{\mathcal{P}}(\bar{b}_{\mathcal{P}})/\bar{b}_{\mathcal{P}}$. As the
budget $B$ increases, the optimal attack $b^{\star}$ activates partitions and drives them
to their caps in order of decreasing $\rho_{\mathcal{P}}$, with at most one partition
partially filled at any budget. For partitions of equal salience $\zeta$, the ordering by
$\rho_{\mathcal{P}}$ coincides with the order of increasing
$\FSSp^{\mathrm{trust}}(\mathcal{P})\,\DWPRm(\mathcal{P})$, so the least substitutable and
least path-robust partitions are attacked first.
\end{proposition}

\begin{proof}
By Proposition~\ref{prop:exist} the payoffs are convex and the maximiser lies at an
extreme point of the feasible polytope $K$, so each active partition is either at its cap
or, for at most one, partially filled by the residual budget. It remains to identify the
priority order. Since the budget is continuous, one partition may be split, so the outer
problem is the fractional relaxation of a knapsack that packs the fixed values
$U_{\mathcal{P}}(\bar{b}_{\mathcal{P}})$ into the budget with weights
$\bar{b}_{\mathcal{P}}$. The fractional knapsack is solved exactly by the greedy rule that
selects items in decreasing value-to-weight ratio $\rho_{\mathcal{P}}=
U_{\mathcal{P}}(\bar{b}_{\mathcal{P}})/\bar{b}_{\mathcal{P}}$, which is the stated
activation order. At the caps every partition sits at the same normalised fraction
$u=1-\epsilon$, so $U_{\mathcal{P}}(\bar{b}_{\mathcal{P}})=\zeta_{\mathcal{P}}
E_{\mathcal{P}}D_{0}\,h(1-\epsilon)$ with the common factor $h(1-\epsilon)$, and hence
$\rho_{\mathcal{P}}=\zeta_{\mathcal{P}}E_{\mathcal{P}}D_{0}h(1-\epsilon)/
\bar{b}_{\mathcal{P}}$. For equal salience the order of decreasing $\rho_{\mathcal{P}}$ is
the order of decreasing $E_{\mathcal{P}}/\bar{b}_{\mathcal{P}}$. Both factors move the same
way with substitutability: $E_{\mathcal{P}}$ is decreasing in
$\FSSp^{\mathrm{trust}}\DWPRm$ by \eqref{eq:expl}, and the cap $\bar{b}_{\mathcal{P}}$ is
smaller for less substitutable partitions because their breakdown threshold $\fmax$ is
lower through Assumption~\ref{ass:break}. Hence decreasing $\rho_{\mathcal{P}}$ is
increasing $\FSSp^{\mathrm{trust}}\DWPRm$, and the least substitutable, least path-robust
partitions are brought to breakdown first.
\end{proof}

\begin{corollary}[Two-partition exchange]\label{cor:two}
Take two partitions with $\rho_{1}\ge\rho_{2}$, that is, partition one has the higher cap
efficiency. The adversary spends its entire budget bringing partition one toward
breakdown until that partition reaches its cap, and only then activates partition two.
The exchange therefore occurs at the budget $b_{\mathrm{th}}=\bar{b}_{1}$, the feasibility
cap of the first partition.
\end{corollary}

\begin{proof}
By Proposition~\ref{prop:conc} the greedy priority order fills partition one to its cap
before partition two receives any budget, because $\rho_{1}\ge\rho_{2}$. Any budget below
$\bar{b}_{1}$ is therefore placed entirely on partition one, and partition two becomes
active exactly when the budget exceeds $\bar{b}_{1}$. Since the convex payoff never
rewards splitting between two partitions before the first is saturated, the switch is
sharp rather than gradual. {This prioritisation reflects the adversary's rational preference for partitions where structural recovery is inherently limited.}
\end{proof}

\subsection{Severity-Breakdown Coupling}
\begin{proposition}[Severity divergence at breakdown]\label{prop:sev}
For every partition, $\Theta_{\mathcal{P}}$ is non-decreasing in the induced fraction
$f_{\mathcal{P}}$ and satisfies $\Theta_{\mathcal{P}}\sim
E_{\mathcal{P}}D_{0}/(1-f_{\mathcal{P}}/\fmax(\mathcal{P}))\to\infty$ as
$f_{\mathcal{P}}\to\fmax(\mathcal{P})^{-}$. Threat severity is therefore governed by
proximity to the defence's Byzantine breakdown threshold, and the divergence rate is set
by the exploitability $E_{\mathcal{P}}$.
\end{proposition}

\begin{proof}
Since $E_{\mathcal{P}}$ does not depend on $b_{\mathcal{P}}$, differentiating
\eqref{eq:tei} through \eqref{eq:sev} and \eqref{eq:drob} gives
$\partial\Theta_{\mathcal{P}}/\partial f_{\mathcal{P}}=E_{\mathcal{P}}D_{0}/
[\fmax(1-f_{\mathcal{P}}/\fmax)^{2}]\ge 0$, so $\Theta_{\mathcal{P}}$ is
monotonically non-decreasing in $f_{\mathcal{P}}$. Substituting
\eqref{eq:drob} into $\Theta_{\mathcal{P}}=E_{\mathcal{P}}\Drob f_{\mathcal{P}}/\fmax$
gives $\Theta_{\mathcal{P}}=E_{\mathcal{P}}D_{0}(f_{\mathcal{P}}/\fmax)/
(1-f_{\mathcal{P}}/\fmax)$, whose leading behaviour as $f_{\mathcal{P}}\to\fmax^{-}$ is
$E_{\mathcal{P}}D_{0}/(1-f_{\mathcal{P}}/\fmax)$, a simple pole at breakdown. This
divergence is the analytic counterpart of the contraction stall in
Theorem~\ref{thm:contraction}, where $\rho(f)\to 1$ and the achieved service, and hence
the inverse degeneracy, collapses.
\end{proof}

\subsection{Stackelberg Configuration and Design-Time Control}
The operator moves first by choosing the macro configuration, then the adversary
best-responds. We record that this leader-follower problem is well posed.

\begin{proposition}[Existence of the Stackelberg configuration]\label{prop:stackel}
Let the operator choose a macro configuration $\theta=(\lamtwo,\FSS,\{\zeta_{\mathcal{P}}\})$
from a compact set $\Theta$, and let the adversary respond with the unique
$b^{\star}(\theta)$ of Proposition~\ref{prop:exist}. Then $b^{\star}(\theta)$ is
continuous in $\theta$, and the operator problem
$\min_{\theta\in\Theta}\TEI(b^{\star}(\theta))$ attains a minimiser.
\end{proposition}

\begin{proof}
The adversary objective is jointly continuous in $(\theta,b)$, and for each $\theta$ the
feasible set $K(\theta)$ of Proposition~\ref{prop:exist} is nonempty, compact, and varies
continuously with $\theta$. Berge's maximum theorem \cite{berge} then gives that the value function is
continuous and the best-response correspondence $b^{\star}(\theta)$ is upper
hemicontinuous, without requiring concavity of the objective. The maximiser is generically
single-valued, since ties in cap efficiency occur only on a measure-zero set of parameters,
so $b^{\star}(\theta)$ is a continuous selection off that set. The composition
$\TEI(b^{\star}(\theta))$ is therefore continuous on the compact set $\Theta$, so it
attains a minimum by the Weierstrass theorem.
\end{proof}

\begin{proposition}[Design-time controllability and exchange rate]\label{prop:control}
Any operator macro action that raises $\fmax(\mathcal{P})$, whether through the
algebraic connectivity $\lamtwo$, through the grouping of functions into partitions,
or through the salience profile, does not increase any
$\Theta_{\mathcal{P}}(b^{\star}_{\mathcal{P}})$ and hence does not increase
$\TEI(b^{\star})$, and strictly decreases them whenever the budget constraint binds.
Moreover, in the low-connectivity regime $\lamtwo\ll 1/c_{3}$, the marginal
substitution between connectivity $\lamtwo$ and substitutability $\FSS$ along an
iso-$\fmax$ contour is constant and equal to
\begin{equation}
\frac{d\lamtwo}{d\FSS}\bigg|_{\fmax}=-\,\frac{c_{2}f^{\star}}{\kappa^{\star}},
\label{eq:exchange}
\end{equation}
where $f^{\star}$ is the contour level and $\kappa^{\star}$ is the maximal value of the
micro factor $\kappa(\beta,\eta_{\tau})$, attained at the interior optimum
$(\beta^{\star},\eta^{\star}_{\tau})$ of Lemma~\ref{lem:kappa}. The negative sign
records that the two levers are substitutes: the
connectivity-substitutability exchange rate, that is, the amount of algebraic
connectivity the operator may give up per unit of additional substitutability at
constant exposure, is $c_{2}f^{\star}/\kappa^{\star}$. In the saturating regime the
local exchange rate is
$\frac{d\lamtwo}{d\FSS}|_{\fmax}=-c_{2}\fmax(1+c_{3}\lamtwo)^{2}/\kappa^{\star}$,
which reduces to \eqref{eq:exchange} as $c_{3}\to 0$.
\end{proposition}

\begin{proof}
By \eqref{eq:sev} and \eqref{eq:drob}, severity $S_{\mathcal{P}}$ is decreasing in
$\fmax(\mathcal{P})$ at fixed $f_{\mathcal{P}}$, since both the ratio
$f_{\mathcal{P}}/\fmax$ and the divergence factor $1/(1-f_{\mathcal{P}}/\fmax)$
decrease. By Proposition~\ref{prop:exist} the best response $b^{\star}$ is a continuous
selection, and raising $\fmax$ lowers the achievable damage per unit budget, so no
$\Theta_{\mathcal{P}}(b^{\star}_{\mathcal{P}})$ increases and $\TEI(b^{\star})$ does not
increase. When the budget binds, some partition sits strictly below breakdown with
positive marginal severity, and lowering that severity strictly lowers the aggregate, so
the decrease is strict. For the exchange rate, hold $\fmax=f^{\star}$ fixed in \eqref{eq:fmax}. With the
linear macro $g(\lamtwo)=\lamtwo$ the contour is
$\kappa^{\star}\lamtwo/\big(1+c_{2}(1-\FSS)\big)=f^{\star}$, that is,
\[
\kappa^{\star}\lamtwo=f^{\star}\big(1+c_{2}\big)-c_{2}f^{\star}\FSS .
\]
Both sides are differentiable in $\FSS$, and $f^{\star}$ and $\kappa^{\star}$ are held
constant along the contour, so differentiating gives
$\kappa^{\star}\,d\lamtwo/d\FSS=-c_{2}f^{\star}$ and hence
$d\lamtwo/d\FSS=-c_{2}f^{\star}/\kappa^{\star}$, which is independent of the operating
point and therefore constant along the contour. For the saturating macro
$g(\lamtwo)=\lamtwo/(1+c_{3}\lamtwo)$ the contour reads
$\kappa^{\star}g(\lamtwo)=f^{\star}\big(1+c_{2}(1-\FSS)\big)$, and the same
differentiation gives $\kappa^{\star}g'(\lamtwo)\,d\lamtwo/d\FSS=-c_{2}f^{\star}$.
Since $g'(\lamtwo)=(1+c_{3}\lamtwo)^{-2}$ and $f^{\star}=\fmax$ on the contour, this
is $d\lamtwo/d\FSS=-c_{2}\fmax(1+c_{3}\lamtwo)^{2}/\kappa^{\star}$, and taking
$c_{3}\to 0$ recovers \eqref{eq:exchange}. The two levers therefore trade off with a
negative slope, and the magnitude of that slope is the exchange rate quoted in the
statement.
\end{proof}

\begin{remark}[Design-time versus runtime]
Proposition~\ref{prop:control} uses the slow macro levers, which gives a design-time
threat model set at deployment. The micro knobs $(\beta,\eta_{\tau})$ act on a faster
timescale through Lemma~\ref{lem:kappa} and give a runtime threat model as the surface
$\Sigma$ drifts. This separation formalises, with a feasibility grounding, the
design-time and runtime split that checklist methods such as STRIDE apply informally.
\end{remark}
{The proposed framework naturally accommodates slowly varying network topologies.
Whenever the coordination graph changes due to node mobility, link failures,
virtual-network-function migration, or infrastructure reconfiguration, the
graph Laplacian and the associated structural metrics are recomputed from the
updated topology. The breakdown threshold, threat scores, and equilibrium attack
allocation are then re-evaluated using the same analytical framework without
modifying the underlying consensus protocol. Consequently, the proposed
formulation is not restricted to static deployments but extends directly to
quasi-static dynamic networks whose topology evolves on a slower timescale than
the consensus convergence process.}
\subsection{Sensitivity and Complexity}
\begin{proposition}[Lipschitz sensitivity]\label{prop:sens}
Suppose the structural metrics are estimated with bounded relative error, so that the
exploitabilities and thresholds used by the solver deviate from their true values by at
most $\delta$ in relative terms. Then the resulting equilibrium allocation and the
Threat Exposure Index deviate from their true values by $O(\delta)$, so the ranking is
stable under small metric misspecification.
\end{proposition}

\begin{proof}
By Proposition~\ref{prop:conc} the equilibrium is the priority-fill allocation determined
by the caps $\bar{b}_{\mathcal{P}}$ and the cap efficiencies $\rho_{\mathcal{P}}$. Both are
smooth functions of the metrics: the cap $\bar{b}_{\mathcal{P}}=f_{\mathcal{P}}^{-1}
((1-\epsilon)\fmax)$ is smooth in $\fmax$ through the invertible map \eqref{eq:finduced},
and $\rho_{\mathcal{P}}$ is smooth in $E_{\mathcal{P}}$, $\Drob$, and $\fmax$. A relative
error of size $\delta$ in the metrics therefore perturbs every cap and efficiency by
$O(\delta)$. Away from ties in $\rho_{\mathcal{P}}$, the priority order is locally constant,
so the allocation moves each filled cap by $O(\delta)$ and moves the single partial spend by
$O(\delta)$ to preserve the budget, giving an $O(\delta)$ change in $b^{\star}$. By
continuity of \eqref{eq:tei} the Threat Exposure Index also changes by $O(\delta)$. The only
non-smooth events are exchanges of adjacent priorities, which occur when two efficiencies
cross; these are confined to a measure-zero set of parameters, and at such a crossing the
two partitions carry nearly equal efficiency so the reordering has vanishing effect on the
index. Rankings, being determined by the ordering of the scores, are preserved
except within an $O(\delta)$ band around ties.
\end{proof}

\begin{proposition}[Complexity]\label{prop:complexity}
The design-time pass solves the budget problem \eqref{eq:adv} over $M$ partitions in
$O(M\log M)$ time once the structural metrics are known, and the metrics themselves are
read from the DC-ABM equilibrium at a cost linear in the number of agents and neighbours.
\end{proposition}

\begin{proof}
By Proposition~\ref{prop:conc} the optimal attack fills partitions to their caps in
decreasing order of cap efficiency $\rho_{\mathcal{P}}$. Computing each cap
$\bar{b}_{\mathcal{P}}$ is a constant-time evaluation of the closed forms
\eqref{eq:fmax} and \eqref{eq:finduced}, since the induced-corruption map is
analytically invertible, and computing each $\rho_{\mathcal{P}}$ is then a
constant-time evaluation of \eqref{eq:expl}, \eqref{eq:drob} and \eqref{eq:tei}. Sorting the $M$ efficiencies
costs $O(M\log M)$, and a single linear sweep then fills the caps in priority order until
the budget is exhausted, at $O(M)$ cost, so the solver runs in $O(M\log M)$. The metric
extraction reads local quantities from the converged consensus, which costs
$O(K\bar{d})$ for $K$ agents of average degree $\bar{d}$, linear in the problem size.
\end{proof}

\subsection{Reference Implementation}\label{sec:impl}
Every result reported in Sec.~\ref{sec:results} is produced by one self-contained
reference implementation with three components, which we describe here so that the
numerical section can be read as a validation of the analysis rather than as a
separate artefact. The first component is an \emph{agent-level simulator}. It instantiates the $K$
defender agents on the coordination graph $\mathcal{G}_{f}$, initialises the states
$x_{k}(0)$ and the reputations $\tau_{kj}(0)$, and integrates the trust-weighted
trimmed-mean update \eqref{eq:consensus} together with the reputation update
\eqref{eq:trust} synchronously. A fraction $f$ of the agents, drawn uniformly at
random, is Byzantine and may emit an arbitrary bounded value to each neighbour at
each step, independently across neighbours and across time. The simulator is the
only place where the consensus dynamics are evaluated; nothing downstream
approximates them.

The second component is a \emph{metric extractor}. It reads the structural
coordinates $\Sigma(\mathcal{P})$ of Definition~\ref{def:surface} from the converged
agent state, one partition at a time, and evaluates the breakdown threshold
$\fmax(\mathcal{P})$ through \eqref{eq:fmax} and the exploitability
$E_{\mathcal{P}}$ through \eqref{eq:expl}. The third component is the \emph{priority-fill equilibrium solver} of
Proposition~\ref{prop:conc}. It forms the caps $\bar{b}_{\mathcal{P}}$ and the cap
efficiencies $\rho_{\mathcal{P}}$, sorts the $M$ partitions by $\rho_{\mathcal{P}}$,
and fills the caps in that order until the budget is exhausted, which is the
$O(M\log M)$ pass of Proposition~\ref{prop:complexity}. Each baseline is implemented against the same graph instances, the same partition
definitions, and the same random seeds, so that every comparison in
Sec.~\ref{sec:results} differs only in the scoring or defence rule and not in the
underlying network realisation. Algorithm~\ref{alg:coevo} summarises the design-time and runtime procedure that results
from the analysis above.

\begin{algorithm}[!t]
\caption{Adversary-as-agents threat modelling}
\label{alg:coevo}
\begin{algorithmic}[1]
\REQUIRE partitions $\Pi$, macro config $(\mathcal{G}_{f},\{\zeta_{\mathcal{P}}\})$,
micro knobs $(\beta,\eta_{\tau})$, budget $B$, alarm level $\TEI^{\mathrm{th}}$
\STATE \textbf{Design time (macro):} run the DC-ABM defence \eqref{eq:consensus} to
equilibrium; for each $\mathcal{P}$ read $\Sigma(\mathcal{P})$ and compute
$\fmax(\mathcal{P})$ via \eqref{eq:fmax}
\STATE compute exploitability $E_{\mathcal{P}}$ \eqref{eq:expl} for all $\mathcal{P}$
\STATE solve the adversary best response \eqref{eq:adv} by filling partitions to their
breakdown caps in decreasing cap-efficiency order to obtain $b^{\star}$
\COMMENT{Prop.~\ref{prop:conc}, $O(M\log M)$}
\STATE evaluate $\Theta_{\mathcal{P}}(b^{\star}_{\mathcal{P}})$ and rank partitions;
report $\TEI(b^{\star})$
\STATE \textbf{Runtime (micro):}
\WHILE{network operating}
  \STATE refresh $\Sigma$ from local agent state; update
  $\{\Theta_{\mathcal{P}}\}$ and $\TEI$
  \IF{$\TEI>\TEI^{\mathrm{th}}$}
    \STATE retune $(\beta,\eta_{\tau})$ one step toward
    $(\beta^{\star},\eta_{\tau}^{\star})$, which raises $\kappa$ and hence
    $\fmax$ \COMMENT{Lemma~\ref{lem:kappa}, \eqref{eq:fmax}}
    \STATE recompute $b^{\star}$ and the ranking
    \IF{$\TEI>\TEI^{\mathrm{th}}$ and the micro levers are exhausted}
      \STATE schedule a macro reconfiguration that raises $\lamtwo$ or $\FSS$
      on the slower deployment timescale \COMMENT{Prop.~\ref{prop:control}}
    \ENDIF
  \ENDIF
\ENDWHILE
\ENSURE ranked threat model $\{\Theta_{\mathcal{P}}\}$, index $\TEI$, recommended levers
\end{algorithmic}
\end{algorithm}

\section{Numerical Results}\label{sec:results}
We instantiate a mobile-network deployment with $M=5$ partitions, namely Byzantine
nodes, jamming, pilot or identity contamination, cache-aware eavesdropping, and
correlated virtualised network function failure. The DC-ABM defence runs on a
coordination graph of $K=60$ entities. Unless stated otherwise the operating point uses
a $16$-regular topology with algebraic connectivity $\lamtwo\approx 2.15$ and
network-wide substitutability $\FSS=0.80$, which by \eqref{eq:fmax} gives a predicted
breakdown threshold $\fmax\approx 0.175$ on the most substitutable partition. The
structural coordinates of the five partitions are listed in Table~\ref{tab:params}. The
consensus, breakdown, and defence results are produced by the agent-level simulator
of Sec.~\ref{sec:impl}, which integrates \eqref{eq:consensus} and \eqref{eq:trust}
directly, and the equilibrium results are produced by the priority-fill solver of
the same implementation. Reported figures average over independent seeds as noted
per experiment. We compare against several state-of-the-art baselines, namely static
attack trees \cite{schneier}, CVSS scoring \cite{cvss}, MulVAL attack graphs
\cite{mulval}, a game-theoretic Stackelberg security model
\cite{gametheory,tambe}, and a multi-agent reinforcement learning red team
\cite{ivoghlian}. For the defended-utility comparison we further include the
alternating direction method of multipliers \cite{reifert}, a multi-agent reinforcement
learning defence with a graph neural network \cite{ivoghlian}, the Su-Vaidya Byzantine
consensus \cite{suvaidya}, a trimmed-mean-only defence, and an undefended baseline.
Table~\ref{tab:sota} first places the threat-modelling baselines beside the proposed
framework along the five capabilities that Secs.~\ref{sec:defence}
and~\ref{sec:coevo} establish, namely an emergent rather than exogenous adversary,
feasibility grounding, substitutability awareness, design-time control, and coupling to
a breakdown threshold. Figs.~\ref{fig:r5},~\ref{fig:r7}, and~\ref{fig:r9} then quantify
the rows of that summary.

\begin{table}[!t]
\centering
\caption{Structural coordinates of the five partitions at the operating point.}
\label{tab:params}
\setlength{\tabcolsep}{4pt}
\renewcommand{\arraystretch}{1.15}
\footnotesize
\begin{tabular}{@{}lccccc@{}}
\toprule
Partition & $\FSSp$ & $\DWPRm$ & $\FSSp\cdot\DWPRm$ & $\zeta$ & $\fmax$ \\
\midrule
Byzantine nodes      & 0.55 & 0.45 & 0.248 & 1.00 & 0.131 \\
Pilot contamination  & 0.65 & 0.55 & 0.358 & 0.80 & 0.146 \\
Correlated VNF       & 0.60 & 0.85 & 0.510 & 0.95 & 0.138 \\
Jamming              & 0.80 & 0.70 & 0.560 & 0.90 & 0.175 \\
Cache eavesdropping  & 0.90 & 0.80 & 0.720 & 0.70 & 0.201 \\
\bottomrule
\end{tabular}
\end{table}

\begin{table}[!t]
\centering
\caption{Qualitative comparison of threat-modelling approaches.
\cmark: supported; \xmark: not supported.}
\label{tab:sota}
\setlength{\tabcolsep}{2.4pt}
\renewcommand{\arraystretch}{1.15}
\scriptsize
\begin{tabular}{@{}lccccc@{}}
\toprule
Approach & \begin{tabular}{@{}c@{}}Emergent\\adversary\end{tabular}
         & \begin{tabular}{@{}c@{}}Feasibility\\grounded\end{tabular}
         & \begin{tabular}{@{}c@{}}Substit.\\aware\end{tabular}
         & \begin{tabular}{@{}c@{}}Design-time\\control\end{tabular}
         & \begin{tabular}{@{}c@{}}Breakdown\\coupled\end{tabular} \\
\midrule
Attack trees \cite{schneier}      & \xmark & \xmark & \xmark & \xmark & \xmark \\
STRIDE/CVSS \cite{shostack,cvss}  & \xmark & \xmark & \xmark & \xmark & \xmark \\
MulVAL graph \cite{mulval}        & \xmark & \xmark & \xmark & \cmark & \xmark \\
Game-theoretic \cite{gametheory}  & \cmark & \xmark & \xmark & \cmark & \xmark \\
Bayesian graph \cite{bayesag}     & \xmark & \cmark & \xmark & \cmark & \xmark \\
MARL red team \cite{ivoghlian}    & \cmark & \xmark & \xmark & \xmark & \xmark \\
\textbf{This work}                & \cmark & \cmark & \cmark & \cmark & \cmark \\
\bottomrule
\end{tabular}
\end{table}

Fig.~\ref{fig:r1} shows the honest disagreement $\Delta(t)$ under the DC-ABM update for
increasing corrupted-mass fractions, averaged over ten seeds. Below the breakdown
threshold the disagreement contracts geometrically, in agreement with
Theorem~\ref{thm:contraction}. The measured per-step factor at $f=0.10$, obtained by
regressing $\log\Delta(t)$ on $t$ over the geometric portion of the trajectory, is
$\hat{\rho}\approx0.95$, that is, the honest spread shrinks by about five percent per
consensus iteration; matching this against \eqref{eq:rate} with
$\fmax\approx0.175$ identifies $\gamma\approx0.12$ for this topology and trust
floor. As $f$ rises toward the predicted threshold the contraction slows and stalls,
and at $f=0.22$, above the threshold for this topology, the honest states no longer
agree. This is the plateau-then-cliff behaviour that the severity result exploits. Fig.~\ref{fig:r2} compares the breakdown threshold predicted by the factorisation
\eqref{eq:fmax} with the threshold measured directly from the dynamics, across seven
topologies of increasing algebraic connectivity and four substitutability levels, that
is, over $28$ configurations. The comparison is informative only if the two sides are
obtained independently, so we state the measurement in full. For a
fixed topology and substitutability
level we sweep the corrupted-mass fraction $f$ by bisection on the interval
$[0,0.5]$. At each trial value we run the agent-level simulator of
Sec.~\ref{sec:impl} from random initial states for the same horizon and over the
same seeds as Fig.~\ref{fig:r1}, and we declare the run converged when the honest
spread $\Delta(t)$ falls below a fixed tolerance within that horizon. The measured
threshold $\hat{f}_{\max}$ is the largest $f$ for which the majority of seeds
converge, and the bisection is stopped once the bracketing interval is narrower than
the sweep resolution. The measured threshold is thus read off the simulated
trajectories alone: the bisection never evaluates \eqref{eq:fmax}, and none of the
terms on its right-hand side, namely $\kappa(\beta,\eta_{\tau})$, $g(\lamtwo)$, and
$c_{2}(1-\FSS)$, enters the measurement. The two sides share only the scalar constants
$c_{2}$ and $c_{3}$, fitted once and held fixed across all $28$ configurations. Two
fitted degrees of freedom against $28$ measured thresholds is what makes
Fig.~\ref{fig:r2} a predictive test of the factorisation rather than a restatement of
it, since a wrong functional form could not be rescued by two constants. The model tracks the
measurement with a root-mean-square error of $0.014$ and a median relative error of
about $8.9$ percent. The saturating macro term of Assumption~\ref{ass:break} fits the
measured thresholds substantially better than a linear term, with residuals of $0.014$
against $0.059$, which confirms the diminishing marginal resilience of added
connectivity and justifies the saturating form.

\begin{figure}[!t]
\centering
\includegraphics[width=\columnwidth]{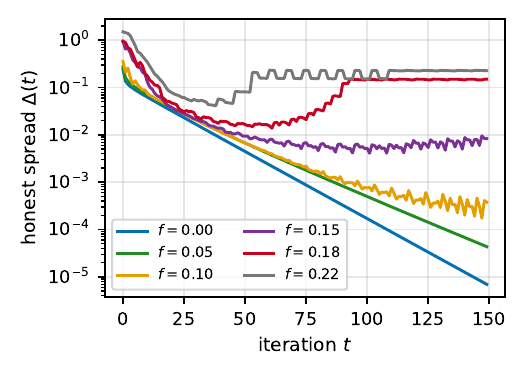}
\vspace{-10mm}
\caption{Consensus contraction of the honest disagreement $\Delta(t)$ under the DC-ABM
defence for increasing corrupted-mass fraction $f$. Below breakdown the disagreement
decays geometrically as predicted by Theorem~\ref{thm:contraction}, and it stalls as
$f$ approaches the threshold.}
\label{fig:r1}
\end{figure}


\begin{figure}[!t]
\centering
\includegraphics[width=\columnwidth]{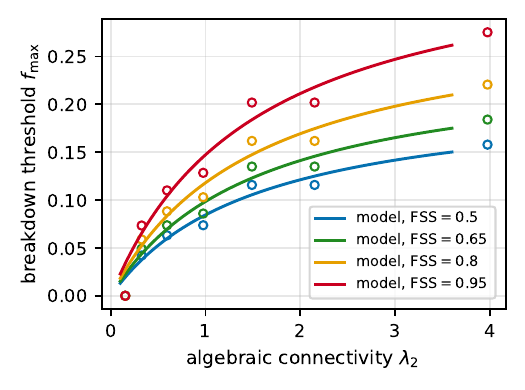}
\vspace{-10mm}
\caption{Breakdown threshold $\fmax$ against algebraic connectivity $\lamtwo$ for four
substitutability levels. Solid curves are the model \eqref{eq:fmax}; open markers are the
empirically measured thresholds. The model tracks the measurement with root-mean-square
error $0.014$.}
\label{fig:r2}
\end{figure}

Fig.~\ref{fig:r3} shows the empirically recovered micro factor
$\hat{\kappa}(\beta,\eta_{\tau})$ as a surface over the trim fraction and the trust rate.
The surface is unimodal, as Lemma~\ref{lem:kappa} predicts, with an interior maximum at
$(\beta^{\star},\eta_{\tau}^{\star})\approx(0.21,0.33)$, which matches the design values
$(0.22,0.30)$ assumed at the operating point. The threshold falls to zero at both ends of
each axis, confirming the two failure modes of trimming too little or too much and of
demoting neighbours too slowly or too aggressively. Fig.~\ref{fig:r4} shows the optimal per-partition spend $b^{\star}_{\mathcal{P}}$ as the
attack budget grows. The partitions activate in a strict order that matches
Proposition~\ref{prop:conc} exactly. Byzantine nodes enter first, followed by pilot
contamination, correlated function failure, jamming, and cache eavesdropping. This is
precisely the order of increasing $\FSSp\cdot\DWPRm$, which runs $0.248$, $0.358$,
$0.510$, $0.560$, $0.720$ across the five partitions, so the emergent attack concentrates
on the least substitutable and least path-robust partitions first.

\begin{figure}[!t]
\centering
\includegraphics[width=\columnwidth]{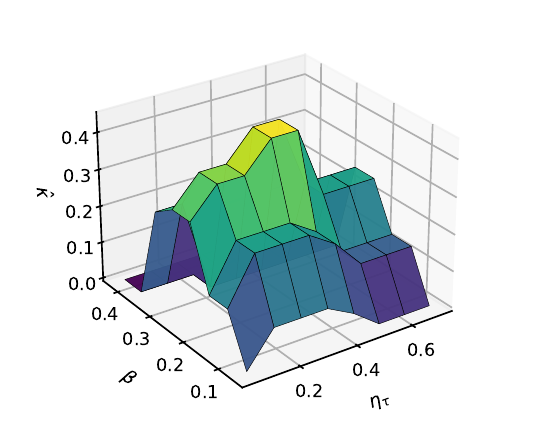}
\vspace{-10mm}
\caption{Empirically recovered micro factor $\hat{\kappa}$ as a function of the trim
fraction $\beta$ and the trust rate $\eta_{\tau}$. The surface is unimodal with an
interior maximum near $(0.21,0.33)$, confirming Lemma~\ref{lem:kappa}.}
\label{fig:r3}
\end{figure}


\begin{figure}[!t]
\centering
\includegraphics[width=\columnwidth]{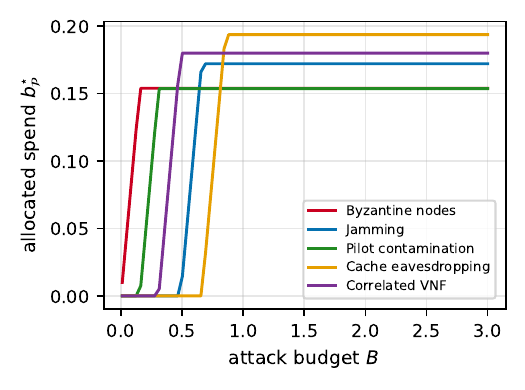}
\vspace{-10mm}
\caption{Optimal attack allocation $b^{\star}_{\mathcal{P}}$ against total budget $B$.
Partitions activate in order of increasing $\FSSp\cdot\DWPRm$, matching the emergent
concentration of Proposition~\ref{prop:conc} exactly.}
\label{fig:r4}
\end{figure}

The experiment reported in Fig.~\ref{fig:r5} tests whether the emergent ranking
predicts where the network actually degrades, using a deliberately non-circular
ground truth. For each of a set of random
instances we build $M=6$ partitions realised as consensus subgraphs of different degree,
so that each partition has a distinct connectivity and therefore a distinct breakdown
threshold. The ground-truth criticality of a partition is the residual consensus
disagreement measured by running the DC-ABM dynamics of
Section~\ref{sec:defence} on that subgraph at the equilibrium attack level produced by
the priority-fill solver. Crucially, this measurement is read from the trajectory of the
averaging process and never uses the exploitability term, the severity term, or the
score $\Theta$, so a correlation between the predicted ranking and the measured
degradation is genuine predictive evidence rather than an identity. We compare the score
against four baselines that see the same instances: a CVSS-style expert severity, an
attack-tree size proxy, node-degree centrality as a representative graph-based method,
and a random ordering. Kendall correlations are averaged over the instances and over a
sweep of attack budgets, and Fig.~\ref{fig:r5} reports the pooled means with standard
errors.

The score of this work attains a mean Kendall correlation of $0.34\pm0.03$ with the
measured degradation, whereas every baseline clusters near or below zero: attack trees
reach $0.07$, CVSS $0.01$, and the random ordering is statistically indistinguishable
from zero. Node-degree centrality is not merely uninformative but strongly
anti-correlated, at $-0.41$, and the reason is structural. In these instances a
partition realised on a higher-degree subgraph has a larger algebraic connectivity,
hence by \eqref{eq:fmax} a higher breakdown threshold, hence a larger margin
$1-f_{\mathcal{P}}/\fmax(\mathcal{P})$ at the equilibrium attack level, and therefore
a \emph{smaller} residual disagreement. Degree centrality ranks exactly those
partitions as the most critical, so it does not simply miss the true ordering, it
reverses it. The separation is the
quantitative content of the paper's thesis in an honest form: endogenising the adversary
produces a ranking that tracks realised degradation, while catalogue and naive-structural
scores that ignore the coupling between substitutability and breakdown do not. The
absolute correlation is moderate rather than near-perfect, and two effects cap it.
First, the ground truth is a noisy scalar read from finite-length consensus
trajectories, so partitions whose true criticalities are close are frequently swapped
by the measurement itself. Second, once the budget is large enough to drive every
partition to its cap, all partitions sit at the same normalised fraction $1-\epsilon$
of their thresholds, the measured residual disagreements become nearly tied, and any
ranking of near-ties is close to arbitrary. What Fig.~\ref{fig:r5} therefore
establishes is the ceiling of an in-model validation, in which both the ranking and
the ground truth are generated by the same simulator; Sec.~\ref{sec:conclusion}
states what a validation on operational topologies with measured attack traces would
add beyond it. Fig.~\ref{fig:r6} compares the defended utility of six defences as the attacked
fraction grows. The defended utility of a defence at attacked fraction
$f_{\mathcal{P}}$ is the achieved-to-target service ratio that the network still
delivers under that defence, that is, the reciprocal $1/D_{k}$ of the feasibility
signal \eqref{eq:feas} averaged over the honest entities at the end of the consensus
run, normalised so that an unattacked network has utility one. A utility of one means
the network still meets its service target; a utility approaching zero means the
consensus no longer supports a usable decision. The DC-ABM defence sustains a utility
of $0.99$ at an attacked fraction of $0.16$,
near its predicted breakdown threshold, and holds the plateau-then-cliff shape that
Theorem~\ref{thm:contraction} predicts. Over the same range the alternating direction
method of multipliers degrades to $0.17$, and the undefended baseline has already failed
below an attacked fraction of $0.10$. The Su-Vaidya and trimmed-mean defences hold an
intermediate plateau but break earlier than DC-ABM, because they lack the trust weighting
that suppresses persistent attackers.

\begin{figure}[!t]
\centering
\includegraphics[width=\columnwidth]{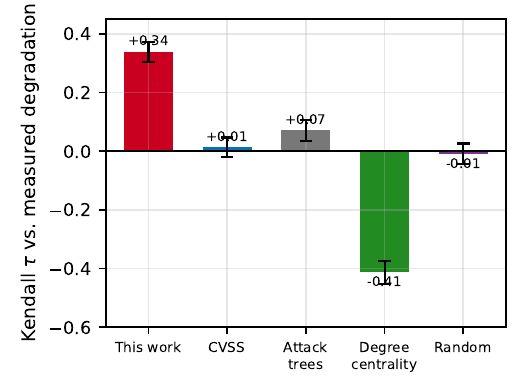}
\vspace{-10mm}
\caption{Ranking fidelity against an independently measured ground truth. Bars are the
Kendall rank correlation between each method's predicted partition criticality and the
consensus degradation measured directly from the defence dynamics, pooled over random
instances and attack budgets with standard-error whiskers. The score of this work
separates clearly from severity-catalogue, attack-tree, graph-centrality, and random
baselines.}
\label{fig:r5}
\end{figure}

\begin{figure}[!t]
\centering
\includegraphics[width=\columnwidth]{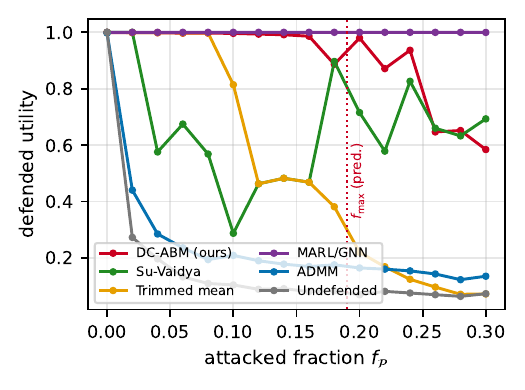}
\vspace{-10mm}
\caption{Defended utility against attacked fraction $f_{\mathcal{P}}$ for six defences.
DC-ABM sustains utility to its predicted threshold, while linear-averaging and undefended
baselines degrade far earlier.}
\label{fig:r6}
\end{figure}

Fig.~\ref{fig:r7} shows the equilibrium Threat Exposure Index against algebraic
connectivity for four substitutability levels, under a binding budget, by which we
mean a total budget $B$ strictly below the sum $\sum_{\mathcal{P}}\bar{b}_{\mathcal{P}}$
of the per-partition feasibility caps of Proposition~\ref{prop:exist}, so that the
constraint $\mathbf{1}^{\top}b\le B$ is active at the optimum and the adversary cannot
bring every partition to its barrier.
Figs.~\ref{fig:r7} and~\ref{fig:r8} fix $B=0.5$, against
$\sum_{\mathcal{P}}\bar{b}_{\mathcal{P}}\approx0.85$ at the operating point of
Table~\ref{tab:params}. This is the regime in which the ordering of
Proposition~\ref{prop:conc} decides which partitions reach breakdown; under a slack
budget every partition would sit at the same normalised fraction $1-\epsilon$ of its
own threshold and the index would be far less sensitive to either lever. Raising
connectivity monotonically lowers the index, as Proposition~\ref{prop:control} requires,
and raising substitutability shifts the whole family downward. Doubling the connectivity
reduces the index by about $19$ percent at fixed substitutability. A connectivity-blind
CVSS baseline is flat across the whole range, since it cannot represent either lever. Fig.~\ref{fig:r8} shows the Threat Exposure Index as a surface over connectivity and
substitutability, with iso-exposure contours projected beneath. The contours trace the
connectivity-substitutability exchange of Proposition~\ref{prop:control}: an operator can
hold the exposure fixed by trading a decrease in one lever for an increase in the other.
The magnitude $c_{2}f^{\star}/\kappa^{\star}$ of the closed-form exchange rate
\eqref{eq:exchange} describes the slope of these contours in the low-connectivity
regime, and the saturating correction steepens the slope as connectivity grows,
consistent with the diminishing marginal resilience seen in Fig.~\ref{fig:r2}.

\begin{figure}[!t]
\centering
\includegraphics[width=\columnwidth]{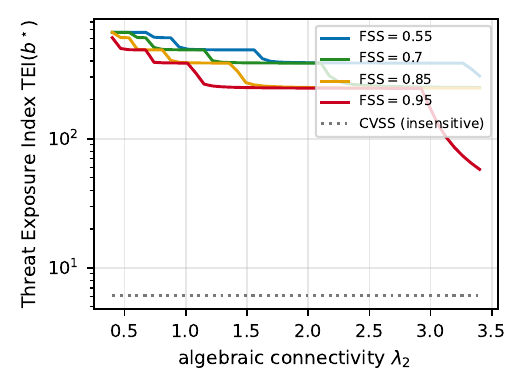}
\vspace{-10mm}
\caption{Threat Exposure Index against algebraic connectivity for four substitutability
levels under a binding budget $B=0.5$, that is, a budget below the cap sum
$\sum_{\mathcal{P}}\bar{b}_{\mathcal{P}}\approx0.85$, so the constraint is active at the
optimum. Both levers lower the index, as
Proposition~\ref{prop:control} predicts, while the CVSS baseline is insensitive.}
\label{fig:r7}
\end{figure}

\begin{figure}[!t]
\centering
\includegraphics[width=\columnwidth]{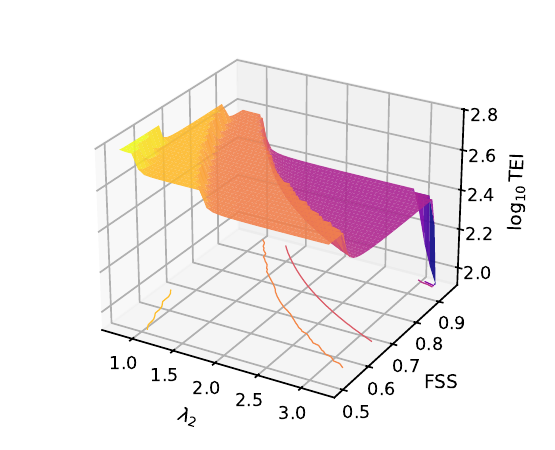}
\vspace{-10mm}
\caption{Threat Exposure Index surface over connectivity and substitutability with
iso-exposure contours beneath, at the binding budget $B=0.5$. Colour encodes surface
height, that is $\log_{10}\TEI(b^{\star})$ on the vertical axis, and carries no
information beyond it: dark blue and purple mark low exposure, orange and yellow mark
high exposure. The projected contours use the same scale, so each contour is a locus of
constant TEI. The contours realise the connectivity-substitutability
exchange of Proposition~\ref{prop:control}.}
\label{fig:r8}
\end{figure}

Fig.~\ref{fig:r9} shows the Threat Exposure Index over time as the substitutability
of the critical partition erodes. This models a drifting mobile environment, for
example one in which user mobility and successive handovers progressively remove the
alternative entities that could serve a function, a virtualised network function is
migrated away from the hosts that carried its redundant replicas, or a link failure
removes a substitute path; in each case $\FSSp^{\mathrm{trust}}$ of the affected
partition falls while the operator makes no change of its own.

The runtime loop of Algorithm~\ref{alg:coevo} responds as follows. At every step the
controller refreshes the surface $\Sigma$ from the local agent state, recomputes the
scores $\{\Theta_{\mathcal{P}}\}$ and the index $\TEI(b^{\star})$, and compares the
index against the operator-set alarm level $\TEI^{\mathrm{th}}$, drawn as the dotted
line in Fig.~\ref{fig:r9}. While the index stays below the alarm level the controller
does nothing, because the drift is still absorbed by the existing margin. The first
time the index crosses the alarm level the controller takes one retuning step of the
micro knobs $(\beta,\eta_{\tau})$ towards the interior optimum
$(\beta^{\star},\eta^{\star}_{\tau})$ of Lemma~\ref{lem:kappa}. That step raises the
micro factor $\kappa$, hence raises $\fmax$ through \eqref{eq:fmax}, hence lowers
every severity term \eqref{eq:sev}, and the equilibrium allocation and the index are
then recomputed on the new configuration. Because the drift is continuous and each
retuning step is bounded, the crossing recurs and the loop fires repeatedly, which is
what produces the staircase in the curve. The controller therefore tracks the drift
instead of correcting it once, holding the index well below the no-retune baseline
based on the alternating direction method of multipliers and achieving about a $37$
percent reduction in the mean index over the drift. A static CVSS assessment is flat and blind to the drift,
so it neither detects nor mitigates the erosion. Fig.~\ref{fig:r10} shows the solve time of the priority-fill equilibrium solver against the number
of partitions on log-log axes, together with an $O(M\log M)$ reference and the superlinear
scaling of a learning red team and of a mixed-integer attack-graph solver. The measured
solver follows the near-linear reference across three orders of magnitude in $M$, in
agreement with Proposition~\ref{prop:complexity}, while the learning and mixed-integer
baselines grow far more steeply and become impractical for large partition counts. The
near-linear cost is what keeps the analysis usable for large mobile deployments with many
attackable modes.

\begin{figure}[!t]
\centering
\includegraphics[width=\columnwidth]{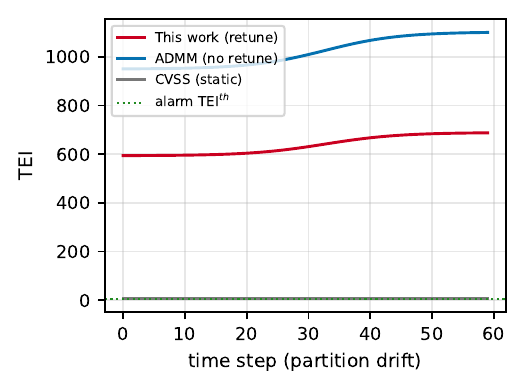}
\vspace{-10mm}
\caption{Threat Exposure Index over time under partition drift. Runtime retuning holds the
index near the alarm level and reduces the mean index by about $37$ percent against a
no-retune baseline, while the static baseline is blind to the drift.}
\label{fig:r9}
\end{figure}

\begin{figure}[!t]
\centering
\includegraphics[width=\columnwidth]{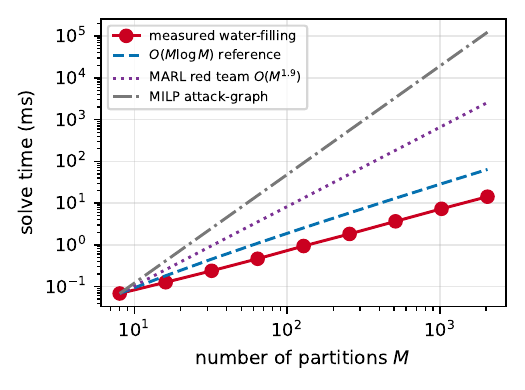}
\vspace{-10mm}
\caption{Solver time against the number of partitions $M$ on log-log axes.
The solver follows the $O(M\log M)$ reference of Proposition~\ref{prop:complexity}, while
learning and mixed-integer baselines scale superlinearly.}
\label{fig:r10}
\end{figure}

Fig.~\ref{fig:r11} shows the rank stability of the emergent threat model, measured as the
Kendall correlation between the ranking under noisy metrics and the ranking under true
metrics, as the relative metric noise grows. The ranking degrades gracefully, and how
gracefully depends on how many coordinates are corrupted at once. Noise injected into a
single metric, whether substitutability, path robustness, or robust degeneracy, gives
three nearly indistinguishable curves that hold above a correlation of $0.8$ up to a
relative noise of about $0.18$ and are still near $0.5$ at $\sigma=0.5$. Perturbing all
three metrics jointly is the worst case and degrades faster, crossing $0.8$ at about
$\sigma=0.12$ and reaching $0.59$ at $\sigma=0.2$. The gap is expected, since the
priority order of Proposition~\ref{prop:conc} is set by the cap efficiency
$\rho_{\mathcal{P}}$, in which independent errors in the three coordinates compound
rather than cancel; the single-metric curves apply when one coordinate is poorly
estimated and the all-metric curve bounds the case in which all of them are. Neither
case shows a threshold at which the ranking collapses. This
is the empirical counterpart of the Lipschitz bound of Proposition~\ref{prop:sens} and
shows that the ranking does not require perfect metric estimation to be useful. Fig.~\ref{fig:r12} shows the per-partition score $\Theta_{\mathcal{P}}$ as a heatmap over
partitions and budget. It visualises the emergent structure directly. The most exploitable
partitions light up first and most intensely as the budget grows, and the ordering across
rows reproduces the concentration order of Fig.~\ref{fig:r4}, giving a compact
operator-facing view of where and when the attack surface fills.

\begin{figure}[!t]
\centering
\includegraphics[width=\columnwidth]{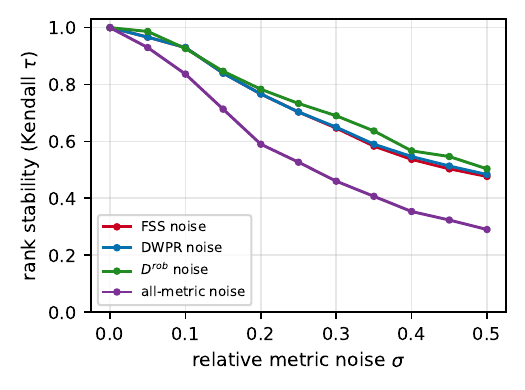}
\vspace{-10mm}
\caption{Rank stability of the emergent threat model against relative metric noise, for
four noise sources. Single-metric noise holds the correlation above $0.8$ up to
$\sigma\approx0.18$; perturbing all three metrics jointly is the worst case and crosses
$0.8$ at $\sigma\approx0.12$. The graceful degradation matches the Lipschitz sensitivity of
Proposition~\ref{prop:sens}.}
\label{fig:r11}
\end{figure}


\begin{figure}[!t]
\centering
\includegraphics[width=\columnwidth]{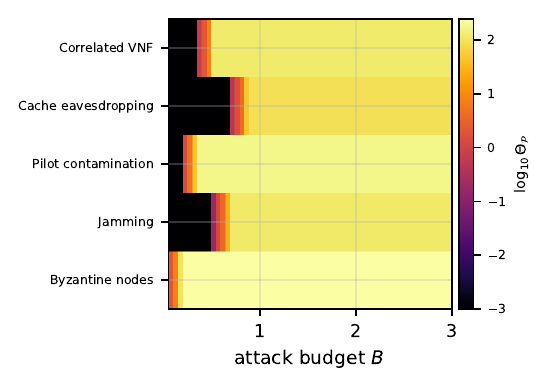}
\vspace{-10mm}
\caption{Per-partition score $\Theta_{\mathcal{P}}$ over partitions and attack budget. The
emergent attack surface fills from the most exploitable partitions outward, reproducing the
concentration order of Fig.~\ref{fig:r4}.}
\label{fig:r12}
\end{figure}

\section{Conclusion}\label{sec:conclusion}
Endogenising the adversary turns wireless and mobile threat modelling from a static
catalogue into an emergent, feasibility-grounded, and operator-controllable analysis. We
gave the explicit decentralised consensus defence, proved that its honest disagreement
contracts below a breakdown threshold with a rate that stalls at the threshold, and proved
that the threshold factorises into a topological and a substitution contribution with a
unimodal micro factor. We then built a feasibility-grounded threat calculus on top of this
defence and established, in order: (i)~the adversary payoff is convex against the
breakdown barrier, so the emergent attack drives a subset of partitions to breakdown
in a strict priority order rather than spreading its budget; (ii)~that priority order
is the order of increasing substitutability, so the emergent attack concentrates on
the least substitutable partitions; (iii)~severity diverges as the induced corrupted
fraction approaches breakdown; (iv)~a Stackelberg configuration of the operator
problem exists; (v)~operators can shrink the emergent attack surface at deployment
time along a closed-form connectivity-substitutability exchange rate; (vi)~the
resulting ranking is Lipschitz stable under misspecification of the structural
metrics; and (vii)~the equilibrium solves in $O(M\log M)$ time. The reference implementation specified in Sec.~\ref{sec:impl}, which integrates the
consensus and trust dynamics directly and closes the loop with the priority-fill
solver, confirmed every result and, in an independent test whose ground truth is the
consensus degradation measured directly from the defence dynamics, showed that the emergent ranking attains a mean Kendall
correlation of about $0.34$ with the measured degradation and outperforms severity-catalogue,
attack-tree, and graph-centrality baselines that all cluster near zero. The moderate absolute
correlation and its dependence on the operating regime mark the current in-model validation as
a first step; a validation on operational topologies with real attack traces is the most
important extension, alongside learning adversaries, mobility-induced non-stationary
partitions, and closing the loop with online macro and micro retuning.

\end{document}